\documentclass[10pt,twocolumn,twoside]{IEEEtran}
\IEEEoverridecommandlockouts 			

\usepackage{hyperref}
\usepackage{graphics} 
\usepackage{epsfig} 
\usepackage{times} 
\usepackage[tbtags]{amsmath} 
\usepackage{amssymb}  
\usepackage{amsthm}
\usepackage{mathtools}
\usepackage{cite}
\usepackage[plain]{algorithm}
\usepackage{algpseudocode}
\usepackage{xcolor}
\usepackage{tikz}
\usepackage{graphicx}

\theoremstyle{plain}
\newtheorem{theorem}{\textbf{Theorem}}
\newtheorem{proposition}{\textbf{Proposition}}
\newtheorem{lemma}{\textbf{Lemma}}
\newtheorem{corollary}{\textbf{Corollary}}

\theoremstyle{definition}
\newtheorem{definition}{Definition}
\newtheorem{example}{\textbf{Example}}
\newtheorem{remark}{\textbf{Remark}}
\newtheorem{assumption}{\textbf{Assumption}}
\usepackage{subcaption}
\usetikzlibrary{shapes}
\usepackage{float}
\usetikzlibrary{calc}
\usepackage{enumerate}

\newcommand{\dsep}{\textrm{d-sep}}
\newcommand{\mean}{\mathbb{E}}

\newcommand{\pa}{\mathcal{P}}
\newcommand{\ch}{\mathcal{C}}
\newcommand\indep{\protect\mathpalette{\protect\independenT}{\perp}}
\def\independenT#1#2{\mathrel{\rlap{$#1#2$}\mkern3mu{#1#2}}}

\title{\LARGE \bf
Causal Structure Identification from Corrupt Data-Streams 
}

\author{Venkat Ram Subramanian, Andrew Lamperski, and Murti V. Salapaka
\thanks{
The authors are with the Department of Electrical and Computer
   Engineering, University of Minnesota, Minneapolis, MN 55455, USA.{\tt\small subra148@umn.edu, alampers@umn.edu,  murtis@umn.edu}}
   \thanks{Work supported in part by NSF CMMI 1727096.}
}

\begin{document}
\maketitle
\begin{abstract}
Complex networked systems can be modeled and represented
as graphs, with nodes representing the agents and the links describing
the dynamic coupling between them. The fundamental objective 
of network identification for dynamic systems is to identify causal
influence pathways. However, 
dynamically related data-streams that originate from different sources
are prone to corruption caused by asynchronous time-stamps, packet
drops, and noise. In this article, we show that identifying causal
structure using corrupt measurements results in the inference of spurious
links. A necessary and sufficient condition that delineates the
effects of corruption on a set of nodes is obtained. Our theory
applies to nonlinear systems, and systems with feedback loops. 
Our results are obtained by the analysis of conditional directed information in
dynamic Bayesian networks.
We provide
consistency results for the conditional directed information
estimator that we use by showing almost-sure
convergence.
\end{abstract} 
\section{Introduction}
Models of systems as networks of interacting systems are central to
many domains such as climate science \cite{kretschmer2016climate}, geoscience \cite{sendrowski2018transfer}, biological systems \cite{omranian2016gene}\cite{bassett2017network}, quantitative finance \cite{fiedor2014networks}, social sciences \cite{borgatti2009network}, and in many engineered systems like the Internet of Things \cite{zhu2015green} and wireless sensor networks \cite{yang2016practical}.
In many scenarios such as the power grid \cite{deka2018structure} and metabolic pathways in cells \cite{finkle2018windowed} it is impractical or impermissible to externally influence the system. Here causal structure identification via passive means is to be accomplished. With advancements in  measurement technology, data processing and communication systems coupled with sensors and measurement devices becoming inexpensive, passive identification of causal graphs of dynamically related agents is becoming more tenable.

Often, the data-streams in such large systems are not immune to effects of noise \cite{stankovic2018distributed}, asynchronous sensor clocks \cite{cho2014survey} and packet drops \cite{leong17sensor}. When dealing with problems of identifying structural and functional connectivity of a large network, there is a pressing need to rigorously study such uncertainties and address detrimental effects of corrupt data-streams on network reconstruction. 
\subsection{Related Work}
Network identification for linear systems is extensively studied. Methods for identifying transfer functions that dynamically link nodes from time-series data are provided in \cite{weerts2018identifiability, hendrickx2018id}, and \cite{materassi2019signal}. However, these works assume that the
time-series are perfect.

Authors in \cite{MatSal12} leveraged  multivariate Wiener filters to reconstruct the undirected topology of the
generative network model.
 Moreover, assuming that the interaction dynamics are {\it strictly causal} and using multivariate estimation based on a Granger filter, it was shown that the interaction structure can be accurately recovered with directions, and without any spurious links.
Here too, results assume data to be uncorrupted with the interaction between agents governed via Linear time-invariant (LTI) dynamics.

For a network of interacting agents with nonlinear, dynamic
dependencies and  strictly causal interactions, the authors in
\cite{quinn15DIG} proposed the use of directed information to
determine the directed structure of the network. Sufficient conditions
to recover the directed structure are provided. More recently,
\cite{sinha2017identifying},\cite{sinha2019information} defined and used \textit{information transfer} to
determine underlying causal interactions in a power network. Here too
it is assumed that the data-streams are ideal with no distortions. 

The authors in \cite{yuan2011robust},\cite{chetty2013robust} use dynamical structure functions (DSF) for network recosntruction \cite{goncalves2008necessary} and consider measurement noise and non-linearities in the network dynamics. The proposed method first finds optimal DSF for all possible Boolean structures and then adopt a model selection procedure to determine the best estimate.
The authors concluded that the performance of their algorithms degrades as noise, network size and non-linearities increase. However, a precise characterization of drawing spurious inferences in structure is not provided. In this article, we provide exact location of spurious links that arise during network reconstruction from corrupt data-streams.

Inspite its significance, little is known on the effects of
uncertainties in the data-streams on network idetnification. Recently in \cite{runge2018causal}, the issues of observation noise and undersampling on causal discovery from time-series data has been addressed. Although authors concluded that spurious links can be inferred, a rigorous characterization of such links was not proven nor a generalization of corruption models was provided. In
\cite{SLS17network} focusing on networks with linear time-invariant (LTI) 
interactions, authors provided characterization of the extent of spurious
links that can appear due to data-corruption. 
However, the analysis is restricted to
LTI systems. Moreover, in \cite{SLS17network} the objective is to determine the topology of the networked system and not to deduce the directions. 
\subsection{Our Contribution}
In this article, we focus our study to determine the directed structure of a network thereby informing the causal structure of the network, using non-invasive means from corrupt data-streams. We consider networks admitting non-linear and strictly causal dynamical interactions. 

We provide necessary and sufficient conditions to determine the directed network structure from corrupt data-streams. We present tight characterization for the spurious links that arise due to corruption of data-streams by determining their location and orientation.

In \cite{SLS18inferring}, preliminary results that characterized the spurious links, in the framework of this article are provided. However, the analysis was limited to dynamical interactions such that every node was dependent dynamically on the entire history (strict) of its \textit{parent} nodes. In this article, we consider general class of non-linear systems by relaxing the above assumption on dynamics. Moreover, we provide detailed and rigorous proofs to genralize the results obtained in \cite{SLS18inferring} wherein only a proof sketch was provided. In addition, we establish convergence results for the estimator that we use to determine conditional directed information.  
 \subsection{Paper Organization}  
We review needed graph theory notions and describe the framework for generative models in Section ~\ref{sec:prelim}. In Section ~\ref{sec:NW pert}, we provide models to characterize corruption of data streams that captures time uncertainty, packet loss and measurement noise. The methods to infer directed network structure for non-linear dynamical systems are described in Section ~\ref{sec:causalID}. Our directed information estimator and simulation results are described in Section ~\ref{sec:DIest}. Finally, a conclusion is provided in Section ~\ref{sec:conclude}.
\section{Preliminaries}\label{sec:prelim}
\subsection{Notations}
\noindent
$y[\cdot ]$ denotes a sequence and $y^{(t)}$ denotes the sequence $y[0],y[1],\dots y[t]$.\\
$P_X$ represents the probability density function of a random variable $X$.\\
$X\indep Y$ denotes that the random variables $X$ and $Y$ are independent.\\
$\mathbb{E}[\cdot ]$ denotes the expectation operator.
\subsection{Graph Theory Background}
In this subsection, few terminologies from graph theory that will be extensively used for network structure inference are reviewed. For furhter reference, see \cite{koller2009prob}.
\begin{definition}[Directed Graph]
\label{def:Graphs}
A \emph{directed graph} $G$ is a pair $(V,A)$ where $V $ is a
set of vertices or nodes and $A$ is a set of edges given by ordered pairs $(i,j)$
where $i,j\in V$. If $(i,j) \in A$, then we say that there is an edge
from $i$ to $j$.
\end{definition}
We shall use $i\to j$ indicates an arc or edge or link from node $i$ to node $j$ in a directed graph. $i - j$ denotes one of $i \to j$ or $j \to i $
\begin{definition}[Children and Parents]
\label{def:ChP}
Given a directed graph ${G}=(V,A)$ and a node $j\in V$, the \emph{children} of $j$ are defined as $\ch (j):=\left\lbrace i|j\to i \in A\right\rbrace $ and the \emph{parents} of $j$ as $\pa (j):=\left\lbrace i|i\to j \in A\right\rbrace $.
\end{definition}
\begin{definition}[Trail/Path]
Nodes $v_1,v_2,\dots ,v_k \in V$ forms a \emph{trail} or a \emph{path} in a directed graph, $G$, if for every $i=1,2,\dots ,k-1$ we have $v_i- v_{i+1}$. 
\end{definition}
\begin{definition}[Chain]
In a directed graph $G$, a \emph{chain} from node $v_i$ to node $v_j$ comprises of a sequence of $k$ nodes such that $v_i\to w_1 \to \dots \to w_{k-2}\to v_j$ holds in $G$.
\end{definition}
\begin{definition}[Descendants and Ancestors]
Suppose there exists a chain from a node $v_j$ to $v_k$ in a directed graph, $G$. Then, $v_k$ is called a \emph{descendant} of node $v_j$ and $v_j$ is called an \emph{ancestor} of $v_k$.
\end{definition}
\begin{definition}[Fork]
A node $v_k$ is a \emph{fork} in a directed graph $G$, if there are two other nodes $v_i,v_j$ such that $v_i\gets v_k \to v_j$ holds. 
\end{definition}
\begin{definition}[Collider]
A node $v_k$ is a \emph{collider} in a directed graph, $G$, if there are two other nodes $v_i,v_j$ such that $v_i\to v_k \gets v_j$ holds. 
\end{definition}
\begin{definition}[Active Trail]
In a directed graph $G$, a trail $v_1- v_2 - \dots - v_n$ is \emph{active} given a set of nodes $Z$ if one of the following statements holds for every triple
$v_{m-1} - v_{m} - v_{m+1}$ along the trail:
\begin{enumerate}[a)]
\item If $v_m$ is not a collider, then $v_m \notin Z$.
\item If $v_m$ is a collider, then $v_m$ or one of its descendants
  is in $Z$.
\end{enumerate} where $m\in \{2,\dots ,n-1\}$.
\end{definition}
See Figure ~\ref{fig:active} for an illustration. 
\begin{figure}[t]
  \centering
  \begin{subfigure}{0.9\columnwidth}
    \centering
 \begin{tikzpicture}[scale=0.35]
                 \tikzstyle{Zuvertex}=[circle,fill=gray,minimum size=12pt,inner sep=0pt,thick,draw]
\tikzstyle{vertex}=[circle,fill=none,minimum size=12pt,inner sep=0pt,thick,draw]        
          \node[vertex] (n1) {$1$};
          
          \node[Zuvertex, right of=n1] (n2) {$2$};
          \node[Zuvertex,right of=n2] (n3) {$3$};
          \node[vertex,right of=n3] (n4) {$4$};          
          \draw[->,thick] (n2)--(n1);
          \draw[->,thick] (n2)--(n3);
                    \draw[->,thick] (n3)--(n4);    
       \end{tikzpicture}   
        \subcaption{\label{fig:actNonCol} Trail connecting 1 and 4 is active given $Z=\{ \}$. }
  \end{subfigure}
    \begin{subfigure}{0.9\columnwidth}
    \centering
        \begin{tikzpicture}[scale=0.35]
                 \tikzstyle{Zuvertex}=[circle,fill=gray,minimum size=12pt,inner sep=0pt,thick,draw]
\tikzstyle{vertex}=[circle,fill=none,minimum size=12pt,inner sep=0pt,thick,draw]        
          \node[vertex] (n1) {$1$};
          
          \node[vertex, right of=n1] (n2) {$2$};
          \node[Zuvertex,right of=n2] (n3) {$3$};
          \node[vertex,right of=n3] (n4) {$4$};          
          \draw[->,thick] (n1)--(n2);
          \draw[->,thick] (n3)--(n2);
                    \draw[->,thick] (n3)--(n4); 
                            \end{tikzpicture}
        \subcaption{\label{fig:actCol}
        Trail connecting 1 and 4 is active given $Z=\{2\}$.
        }
  \end{subfigure}
         \caption{
    \label{fig:active} This figure shows when the trail connecting nodes 1 and 4 is active given $Z$.   
  }
\end{figure}
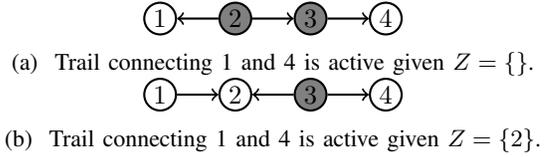
\begin{definition}[d-separation]
Let $X,Y$ and $Z$ be a set of nodes in a directed graph, $G$. In $G$, $X$ and $Y$ are \emph{d-separated} by $Z$ if and only if there is no active trail between any $x\in X$ and any $y\in Y$ given $Z$. It is denoted as $\dsep\left(X,Y\mid Z \right)$.
\end{definition}
\begin{definition}[Directed Cycle]
A \emph{directed cycle} from a node $v_i$ to  $v_i$ in a directed graph, $G$, has the form $v_i\to w_1 \to \dots \to w_k \to v_i$ for some set of nodes $\{w_n\}_{n=1}^{k}$ in $G$.
\end{definition}
\begin{definition}[Directed Acyclic Graph]
A directed graph with no directed cycles is called a \emph{directed acyclic graph} (DAG).
\end{definition}
\begin{definition}[Bayesian Network]
Suppose $G=(V,A)$ is a DAG whose $N$ nodes represent random variables $a_1,\dots,a_N$. $G$ is  called a {\it Bayesian Network}(BN) if for any three subsets $X,\ Y$ and $Z$ of $V$, 
$\mbox{d-sep}(X,Y\mid Z)$ implies $X$ is independent of $Y$ given $Z.$
\end{definition}
\begin{definition}[Faithful Bayesian network]
Suppose $G=(V,A)$ is a DAG whose $N$ nodes represent random variables $a_1,\dots,a_N$. $G$ is called a \emph{Faithful Bayesian network} if for any three subsets $X,\ Y$ and $Z$ of $V$, it  holds that $X$ and $Y$ are independent given $Z$, if and only if  $\mbox{d-sep}(X,Y\mid Z)$ is true.  
\end{definition}
\subsection{Generative Model}\label{sec:genmodel}
In this subsection, the \textit{generative model} that is assumed to generate the
measured data is described. Consider $N$ agents that interact over a
network. For each agent $i$, we associate a discrete time sequence $y_i[\cdot]$
and a sequence $e_i[\cdot].$ We assume $e_i$ and $y_i$ to be random processes. The process $e_i[\cdot]$ is considered
innate to agent $i$ and thus $e_i$ is independent of $e_j$ if $i\not=
j.$ Moreover, $e_i$ is considered to be uncorrelated across time. Let $Y$ denote the set of all random process $\{y_1,\ldots, y_N\}$ with a parent set $\mathcal{P}' (i)$ defined for $i=1,\ldots,N.$  
We consider strictly causal non-linear dynamical relations. Here, $i $ can belong to the parent set $\mathcal{P}'(i)$. The generative model takes the form:
\begin{equation}\label{eq:genmodel}
y_i[t]=f_i\left(y_i^{(t-1)},\underset{j\in \mathcal{P}' (i)}{\bigcup} y_j^{(t-1)},e_i[t]\right),
\end{equation}
where $f_i$'s can be any finite valued non-linear function such that $|f_i|< \infty$.  

For an illustration, consider the dynamics of a generative model described by:
\begin{equation}\label{eq:eggenmodel}
\begin{aligned}
y_1[t]&=y_1[t-1]y_1[t-2]+e_1[t],\\
y_2[t]&=\sin (y_1[t-1]\cdot y_2[t-1]+e_2[t]),\\
y_3[t]&=(y_1[t-1]+y_3[t-1])\cdot e_3[t],\\
y_4[t]&=y_2[t-1]^2+y_3[t-2]+y_4[t-1]+e_4[t],\\
y_5[t]&=y_5[t-1]\cdot y_4[t-1]+e_5[t].
\end{aligned}
\end{equation}
We remark that if $y_j$ appears on the right hand side of ~\eqref{eq:genmodel} for any time instant $t$, then $j\in \mathcal{P}'(i)$; the parent set is thus not dependent on time. 

\begin{figure}
\centering
\subcaptionbox{
   \label{fig:gengraph} 
    Generative Graph $G$} 
 {\begin{subfigure}{0.49\columnwidth}
\centering
 \begin{tikzpicture}[scale=0.65]
 \tikzstyle{vertex}=[circle,fill=none,minimum size=10pt,inner sep=0pt,thick,draw]
  \node[vertex] (n1) {$1$};

  \node[vertex] (n2) at ($(n1) - (2em,3em)$) {$2$};%

        \node[vertex] (n3) at ($(n1) - (-2em,3em)$) {$3$};%

    \node[vertex] at ($(n1)-(0,6em)$) (n4) {$4$};%
	\node[vertex] at ($(n4)-(0,3.5em)$) (n5){$5$};
  \draw[->,thick] (n1)--(n2);
  \draw[->,thick] (n1)--(n3);
  \draw[->,thick] (n2)--(n4);
  \draw[->,thick] (n3)--(n4);
  \draw[->,thick] (n4)--(n5);
\end{tikzpicture}
     \end{subfigure}}
\subcaptionbox{
   \label{fig:dbn} 
    DBN $G'$ for 3 time slices}
     {\begin{subfigure}{0.49\columnwidth}
  \centering
 \begin{tikzpicture}[scale=0.8]
 \tikzstyle{vertex}=[fill=none,minimum size=11.5pt,inner sep=0pt,thin,draw]
 \tikzstyle{pvertex}=[star,star points=10,fill=white,minimum size=10pt,inner sep=0pt,thick,draw]
  \node[vertex] (n1) {$y_1[0]$};

  \node[vertex] (n2) at ($(n1)-(0,3.75em)$) {$y_2[0]$};%

        \node[vertex] (n3) at ($(n2)-(0,3.75em)$) {$y_3[0]$};%

    \node[vertex] at ($(n3)-(0,3.75em)$) (n4) {$y_4[0]$};%
	\node[vertex] at ($(n4)-(0,3.75em)$) (n5){$y_5[0]$};
	
	  \node[vertex] (n6) at ($(n1)-(-3.75em,0)$) {$y_1[1]$};

  \node[vertex] (n7) at ($(n6)-(0,3.75em)$) {$y_2[1]$};%

        \node[vertex] (n8) at ($(n7)-(0,3.75em)$) {$y_3[1]$};%

    \node[vertex] at ($(n8)-(0,3.75em)$) (n9) {$y_4[1]$};%
	\node[vertex] at ($(n9)-(0,3.75em)$) (n10){$y_5[1]$};
	\node[vertex] (n11) at ($(n6)-(-3.75em,0)$) {$y_1[2]$};

  \node[vertex] (n12) at ($(n11)-(0,3.75em)$) {$y_2[2]$};%

        \node[vertex] (n13) at ($(n12)-(0,3.75em)$) {$y_3[2]$};%

    \node[vertex] at ($(n13)-(0,3.75em)$) (n14) {$y_4[2]$};%
	\node[vertex] at ($(n14)-(0,3.75em)$) (n15){$y_5[2]$};
	
  \draw[->,thick] (n1)--(n6);
  \draw[->,thick] (n2)--(n7);
  \draw[->,thick] (n3)--(n8);
  \draw[->,thick] (n4)--(n9);
  \draw[->,thick] (n5)--(n10);
  
  \draw[->,thick] (n1)--(n6);
  \draw[->,thick] (n1)--(n7);
  \draw[->,thick] (n1)--(n8);
  \draw[->,thick] (n2)--(n9);
  \draw[->,thick] (n4)--(n10);
  
  \draw[->,thick] (n6)--(n11);
  \draw[->,thick] (n7)--(n12);
  \draw[->,thick] (n8)--(n13);
  \draw[->,thick] (n9)--(n14);
  \draw[->,thick] (n10)--(n15);
  
  \draw[->,thick] (n6)--(n11);
  \draw[->,thick] (n6)--(n12);
  \draw[->,thick] (n6)--(n13);
  \draw[->,thick] (n7)--(n14);
  \draw[->,thick] (n9)--(n15);
  
  \draw[->,thick] (n1) to[out=35,in=145] (n11);
  \draw[->,thick] (n3)--(n14);
  
  
\end{tikzpicture}
     \end{subfigure}}
    \caption{\label{fig:gendbn}This figure shows \ref{fig:gengraph} generative graph, \ref{fig:dbn} its associated DBN for 3 time slices. }
    \end{figure}
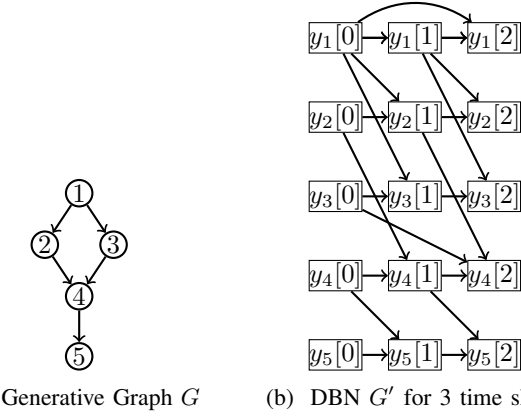
\subsection{Graphical Representation}
Here we describe how networks of dynamical systems are represented by graphs. 
\subsubsection*{Generative Graph}\label{sec:gengraph}
The structural description of \eqref{eq:genmodel} induces a {\it generative graph} $G=(V,A)$ formed by identifying each vertex $v_i$ in $V$ with random process $y_i$ and the set of directed links, $A,$ obtained by introducing a directed link from every element in the parent set $\mathcal{P}'(i)$ of agent $i$ to $i.$ Note that we do not show $i\to i$ in the generative graph and neither do we show the processes $e_i$.

The generative graph associated with the example described in ~\eqref{eq:eggenmodel} is given by Fig. \ref{fig:gendbn}(a). Note that the generative graph describes the relationships between the stochastic processes $y_i$. When the time variable is unraveled we obtain the Dynamic Bayesian Network as defined below.
\subsubsection*{Dynamic Bayesian Network (DBN)}\label{sec:dbn}
Let $G=(V,A)$ be a generative graph. Let $y_i$ be as defined in \eqref{eq:genmodel} for all $ i \in V$. Suppose all discrete time sequences have a finite horizon assumed to be $T$. Let $S_{ij}[t]=\{t':y_j[t'] \in y_j^{(t-1)} \textrm{as an argument of } f_i \textrm{ in expression of } y_i[t] \textrm{ in ~\eqref{eq:genmodel}}\}$ for all $j\in \mathcal{P}'(i)\cup \{i\}$ and for all $t$. Consider the graph $G'=(V',A')$ where $V'= \left( \underset{t\in \{0,1,\dots T\}}{\underset{i\in V}{\bigcup}} y_i[t] \right)$ and 

\noindent$
A'= \underset{t\in \{0,1,\dots T\}}{\underset{i\in V}{\bigcup}} \left( \underset{j\in \mathcal{P}'(i)\cup \{i\}}{\bigcup}\left( \underset{k\in S_{ij}[t]}{\bigcup} y_j[k]\to y_i[t]\right)\right)$

\noindent The joint distribution of $Y^{(T)}$ is given by:
\begin{equation}\label{eq:jdbn}
P_{Y^{(T)}}=P_{y_1[0]}\dots P_{y_N[0]}\prod _{t=1}^T\prod _{i=1}^NP_{y_i[t]\mid \pa(y_i[t])},
\end{equation}
where the parents of $y_i[t]$ are obtained from $G'.$  It can be shown that  $G'$ is the Bayesian network for the random variables $\{y_i[t]: t=0,1,2,\dots , T,\ i=1,2,\dots ,N\}$ and is  considered the \emph{Dynamic Bayesian Network} for $\{y_i: i=1,2,\dots ,N \}$(see \cite{koller2009prob}).
Figure 2(b) represents the DBN for the system in \eqref{eq:eggenmodel} for three time steps.
\section{Uncertainty Description }\label{sec:NW pert}
In this section we provide a description for how uncertainty affects the time-series $y_i.$ We interchangeably use corruption or perturbation to denote uncertainties in daata-streams.
\subsection{General Perturbation Models}\label{sec:Models}
Consider $i^{th}$ node in a generative graph and it's associated unperturbed time-series $y_i$. The corrupt data-stream $u_i$ associated with $i$ follows:
\begin{equation}\label{eq:corruptionModel}
u_i[t]=g_i(y_i^{(t)},u_i^{(t-1)},\zeta_i[t ]),
\end{equation}
where  $u_i$ can depend dynamically on  $y_i$ till time $t$, its own values in the strict past,  and $\zeta_i[t]$ which represents
a stochastic process that is independent across time. 
We highlight a few important perturbation models that are practically relevant. See \cite{SLS18inferring} for more details.
\subsubsection*{Temporal Uncertainty}
Consider a node $i$ in a generative graph. Suppose $t$ is the true clock index but the node $i$ measures a noisy clock index which is given by a random process, $\zeta _i[t]$. One such probabilistic model is given by the following IID Bernoulli process: 
\begin{equation*}
\zeta _i[t]=\begin{cases}
d_1, & \textrm{ with probability } p_i \\
d_2, & \textrm{ with probability } (1-p_i),
\end{cases}
\end{equation*}
where $d_1$ and $d_2$ are any non-positive integers such that at least one of $d_1$ and $d_2$ are not equal to $0$. 
Randomized delays in information transmission can be modeled as a convolution operation with the impulse function $\delta [t]$ shifted by $\zeta _i[t]$ as follows :
\begin{equation}
\label{eq:randDelayMdl}
u_i[t] =\delta [t+\zeta _i[t]]*y_i[t],
\end{equation}
where,
\begin{equation*}
\delta [t]=\begin{cases}
1, & t=0 \\
0, & t\neq 0.
\end{cases}
\end{equation*}
\subsubsection*{Noisy Filtering}
Given a node $i$ in a generative graph, the data-stream $y_i$ is
causally filtered and corrupted
with independent measurement noise $\zeta _i[\cdot ]$. This perturbation model is described by:
\begin{equation}\label{eq:noise}
u_i[t] = (L_i * y_i)[t] + \zeta _i[t],
\end{equation}
where $L_i$ is a stable causal linear time invariant filter. 
\subsubsection*{Packet Drops}
The measurement $u_i[t]$ corresponding to an ideal data-point $y_i[t]$ packet reception at time $t$ can be stochastically modeled as: 
\begin{equation}
  \label{eq:packetDrop}	
u_i[t]=\begin{cases}
y_i[t], & \textrm{ with probability } p_i\\
u_i[t-1], & \textrm{ with probability } (1-p_i).
\end{cases}
\end{equation}
Consider an IID Bernoulli process $\zeta_i$ described by,
\begin{equation*}
\zeta _i[t]=\begin{cases}
1, & \textrm{ with probability } p_i \\
0, & \textrm{ with probability } (1-p_i).
\end{cases}
\end{equation*}
The corruption model in \eqref{eq:corruptionModel} takes the form: 
\begin{equation}
u_i[t]=\zeta _i[t]y_i[t]+(1-\zeta _i[t])u_i[t-1].
\end{equation}
\begin{figure}
\centering
\begin{minipage}{0.9\columnwidth}
  \centering
 \begin{tikzpicture}[scale=0.75]
 \tikzstyle{vertex}=[fill=none,minimum size=10pt,inner sep=0pt,thin,draw]
  \tikzstyle{Zuvertex}=[fill=gray,minimum size=10pt,inner sep=0pt,thin,draw]
  \node[Zuvertex] (n1) {$y_1[0]$};
  \node[vertex] (n16) at ($(n1)-(0,-3.75em)$) {$u_1[0]$};
    \node[vertex] (n17) at ($(n16)-(-3.75em,0)$) {$u_1[1]$};
      \node[vertex] (n18) at ($(n17)-(-3.75em,0)$) {$u_1[2]$}; 
  \node[vertex] (n2) at ($(n1)-(0,3.75em)$) {$y_2[0]$};%

        \node[vertex] (n3) at ($(n2)-(0,3.75em)$) {$y_3[0]$};%

    \node[vertex] at ($(n3)-(0,3.75em)$) (n4) {$y_4[0]$};%
	\node[vertex] at ($(n4)-(0,3.75em)$) (n5){$y_5[0]$};
	
	  \node[Zuvertex] (n6) at ($(n1)-(-3.75em,0)$) {$y_1[1]$};

  \node[vertex] (n7) at ($(n6)-(0,3.75em)$) {$y_2[1]$};%

        \node[vertex] (n8) at ($(n7)-(0,3.75em)$) {$y_3[1]$};%

    \node[vertex] at ($(n8)-(0,3.75em)$) (n9) {$y_4[1]$};%
	\node[vertex] at ($(n9)-(0,3.75em)$) (n10){$y_5[1]$};
	\node[Zuvertex] (n11) at ($(n6)-(-3.75em,0)$) {$y_1[2]$};

  \node[vertex] (n12) at ($(n11)-(0,3.75em)$) {$y_2[2]$};%

        \node[vertex] (n13) at ($(n12)-(0,3.75em)$) {$y_3[2]$};%

    \node[vertex] at ($(n13)-(0,3.75em)$) (n14) {$y_4[2]$};%
	\node[vertex] at ($(n14)-(0,3.75em)$) (n15){$y_5[2]$};
	
  \draw[->,thick] (n1)--(n6);
  \draw[->,thick] (n2)--(n7);
  \draw[->,thick] (n3)--(n8);
  \draw[->,thick] (n4)--(n9);
  \draw[->,thick] (n5)--(n10);

  \draw[->,red,dashed] (n1)--(n16);  
    \draw[->,red,dashed] (n6)--(n17);
      \draw[->,red,dashed] (n11)--(n18);
         \draw[->,red,dashed] (n1)--(n17);
               \draw[->,red,dashed] (n1)--(n18);
                     \draw[->,red,dashed] (n6)--(n18);   
  
  \draw[->,thick] (n1)--(n6);
  \draw[->,thick] (n1)--(n7);
  \draw[->,thick] (n1)--(n8);
  \draw[->,thick] (n2)--(n9);
  \draw[->,thick] (n4)--(n10);
  
  \draw[->,thick] (n6)--(n11);
  \draw[->,thick] (n7)--(n12);
  \draw[->,thick] (n8)--(n13);
  \draw[->,thick] (n9)--(n14);
  \draw[->,thick] (n10)--(n15);
  
  \draw[->,thick] (n6)--(n11);
  \draw[->,thick] (n6)--(n12);
  \draw[->,thick] (n6)--(n13);
  \draw[->,thick] (n7)--(n14);
  \draw[->,thick] (n9)--(n15);
  
  \draw[->,thick] (n1) to[out=35,in=145] (n11);
  \draw[->,thick] (n3)--(n14);

  
\end{tikzpicture}
 \caption{
   \label{fig:pdbn} 
    Perturbed DBN $G'_Z$ for 3 time slices when node $1$ is corrupt. Node 1 ideal stream denoted by $y_1$ is shaded because it is not observed or measured. Only, the its corrupted data-stream, $u_1$, is measured.}
    \end{minipage}
    \end{figure}
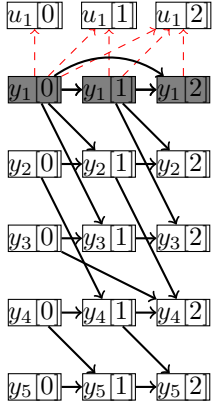    
\subsection{Perturbed Dynamic Bayesian Network}
Here, we provide a discussion on how the dynamic Bayesian network associated with the measured data-streams gets altered when the data-streams are subject to corruption. Consider a generative graph $G=(V,A)$. Let $y_i$ be as defined in \eqref{eq:genmodel} for all $ i \in V$. Suppose all discrete time sequences have a finite horizon assumed to be $T$. Let $G'=(V',A')$ be the associated dynamic Bayesian network. Suppose $Z \subset V$ is the set of perturbed nodes with perturbation model described in \eqref{eq:corruptionModel}. For $i \in Z$, the measured(corrupt) data-stream corresponding to agent $i$, $u_i$, is related to
$y_i$ via \eqref{eq:corruptionModel}. Let $U_Z=\{u_i\}_{i\in Z}$ and $Y_{\bar Z}=\{y_j\}_{j\in \bar Z}$ where $\bar Z=V\setminus Z$. Due to corruption only $U_Z$ and $Y_{\bar Z}$ are measured and observed. Denote the measured data-streams by $\mathcal{W} =U_Z\cup Y_{\bar Z}$. For all $j\in Z$ let $SU_{j}[t]=\{t':u_j[t'] \in u_j^{(t-1)}\textrm{ as an argument of } g_i  \textrm{ in expression of } u_j[t]\textrm{ in ~\eqref{eq:corruptionModel}}\}$ and let $SY_{j}[t]=\{t':y_j[t'] \in y_j^{(t)}\textrm{ as an argument of } g_i \textrm{ in expression of } u_j[t] \textrm{ in ~\eqref{eq:corruptionModel}}\}$ for all $t$. Consider the graph $G'_Z =(V'_Z,A'_Z)$ where  $V'_Z= V' \cup \left( \underset{t\in \{0,1,\dots T\}}{\underset{k\in Z}{\bigcup}} u_k[t] \right)$ and   $A' _Z= A' \cup \left( \underset{i\in SY_k[t]}{\underset{k\in Z}{\bigcup}}y_k[i]\to u_k[t]\right) \cup \left( \underset{i\in SU_k[t]}{\underset{k\in Z}{\bigcup}} u_k[i] \to u_k[t] \right)$ for all $t \in \{0,1,2,\dots , T\}$. Note that the vertex set $V'_Z$ consists of all measurements given by the set $\mathcal{W}$, and the uncorrupted versions $y_k$ of the corrupted versions $u_k$ for $k\in Z.$

Consider the set of random variables, $R=\{y_i[t]:i\in \{1,2,\dots , N\} \mbox{ and } t\in  \{0,1,2,3,\dots , T\}\}\cup \{u_i[t]:i\in \{1,2,\dots , N\} \mbox{ and } t\in \{0,1,2,3,\dots , T\}\} $. The joint distribution $P_R $ is given by:
\begin{multline}\label{eq:jpdbn}
P_{R}=\left(\prod _{i\in V}P_{u_i[0]}\right)\cdot \left(\prod _{j\in Z}P_{y_j[0]}\right) \cdot \left(\prod _{t=1}^T\prod _{i=1}^N P_{u_i[t]\mid \pa(u_i[t])}\right)
\\
\cdot \left(\prod _{t=1}^T\prod _{j=1}^NP_{y_j[t]\mid \pa(y_j[t])}\right),
\end{multline}
where the parents of $u_i[t],y_j[t]$ are obtained from $G'_Z.$ $G'_Z$ is the Bayesian Network for the random variables $R$ and is considered as the perturbed DBN (PDBN) associated with $U_Z\cup Y$.
 
Fig ~\ref{fig:pdbn}. shows an example of a perturbed DBN corresponding to the generative graph in Fig. ~\ref{fig:gendbn}(a) for three time slices when node 1 data-streams are corrupt following a noisy filtering model described in \eqref{eq:noise}.
\section{Structure Identification}\label{sec:causalID}
\subsection{Structure Inference from Ideal Data-Streams}
First, we recall how
the structure of a generative graph can be inferred using directed information in the case of ideal data-streams. 
Consider a generative graph $G$ with $N$ nodes and let $Y$ denote the collection of $N$ data-streams that are measured. The authors in \cite{quinn15DIG} defined and applied directed information (DI) in a network of of dynamically interacting agents, to determine if a process causally influences another. A slightly modified definition of  DI as defined in \cite{quinn15DIG} is: 
\begin{definition}[Directed Information]  
The directed information (DI) from data-stream $y_j$ to $y_i$ is given by:
\begin{equation}\label{eq:di}
I(y_j \to y_i \parallel Y_{\bar{i}\bar{j}})=\mean \left[ \log \frac{P_{y_i\parallel y_j,Y_{\bar{i}\bar{j}}}}{P_{y_i\parallel Y_{\bar{i}\bar{j}}}}\right],
\end{equation} 
where $P_{y_i\parallel y_j,Y_{\bar{i}\bar{j}}}= \overset{T}{\underset{t= 1}{\prod}} P_{y_i[t]\mid y_i^{(t-1)},y_j^{(t-1)},Y_{\bar{i}\bar{j}}^{(t-1)}}$, $P_{y_i\parallel Y_{\bar{i}\bar{j}}}= \overset{T}{\underset{t= 1}{\prod}} P_{y_i[t]\mid y_i^{(t-1)},Y_{\bar{i}\bar{j}}^{(t-1)}}$ and $Y_{\bar{i}\bar{j}}=Y\setminus \{y_i,y_j\}$.
\end{definition}
\noindent Note that DI is always non-negative. So, if there is no directed edge
from $j$ to $i$ in $G$, then we must have that $I(y_j\to
y_i \parallel Y_{\bar{i}\bar{j}}) = 0$.
 
The following theorem was proved in \cite{quinn15DIG} that specifies a
necessary and sufficient condition to detect a presence of link in the generative 
graph.
\begin{theorem}\label{thm:DIG}
  {\it
A directed edge from $j$ to $i$ exists in the directed graph $G$ if
and only if $I(y_j \to y_i \parallel Y_{\bar{i}\bar{j}})>0.$
}
\end{theorem}
\begin{remark}\label{rem:positivity}
In \cite{quinn15DIG}, the authors assume positive distribution for the random processes in $Y$. Under this assumption the result in Theorem \ref{thm:DIG} is both necessary and sufficient. 
\end{remark}
\subsection{Main Result: Inferring Directed Graphs from Corrupt Data-streams}
In this subsection, we will describe how data uncertainty
will lead to spurious probabilistic relationships between nodes that are not
connected in the original graph.

To present the main result in Theorem ~\ref{thm:PG}, some definitions are required. 
\begin{definition}[Perturbed Graph]
\label{def:pg}
Let ${G}=(V,A)$ be a generative graph. Suppose $Z \subset V$ is the
set of perturbed nodes with each perturbation model admitting a description provided in
\eqref{eq:corruptionModel}. The perturbed graph, $G_Z = (V,A_Z)$, is a directed
graph where there is an edge $i\to j \in A_Z$ if and only if there
is a trail, $trl_G: i=v_1 - v_2 - \cdots - v_{k-1} - v_k = j$ in $G$ such that
the following conditions hold:
\begin{enumerate}[P1)]
\item If $j\notin Z$, then $v_{k-1} \to j \in A$.
  \label{case:end}
\item For $m\in \{2,3,\ldots ,k-1\}$, if $v_{m-1} \to v_{m} \leftarrow v_{m+1}$, and $v_m \notin
  Z$, then $v_{m+1} \in Z$.
  \label{case:collider}
\item If $v_m$ is a node such that $v_{m-1}-v_m-v_{m+1}$ is a sub-path of the path $v_1-\ldots-v_k$ and $v_m$ is not a collider, then $v_m\in Z.$ 
  \label{case:notCollider}
\end{enumerate}
\end{definition}  
\begin{remark}
\label{rem:otherTrails}
  Note that the existence of a trail that does not meet the `if' conditions in P\ref{case:end}), P\ref{case:collider}) and P\ref{case:notCollider}) 
  guarantees that $i\to j \in A_Z$.  
  For example, if $i\to j \in A$
  then $i\to j \in A_Z$. Indeed, if $j\notin Z$ then $i\to j\in A_Z$ by
  condition P\ref{case:end}).Conditions P\ref{case:collider}) and P\ref{case:notCollider}) are not applicable. On the other
  hand, if $j\in Z$, then none of the conditions P\ref{case:end}), P\ref{case:collider}) or P\ref{case:notCollider}) are applicable to the trail $i\rightarrow j$. So, $i\to j\in A_Z$.
 \end{remark}
\begin{definition}[Spurious Links]
\label{def:spuriouslinks}
Let $G=(V,A)$ be a generative graph, $Z \subset V$ be the set of perturbed nodes and $G_Z=(V,A_Z)$ be the perturbed graph. Spurious links are those links $i\to j \in A_Z$ that do not belong to $A$.
\end{definition}
The following lemma will be used to prove our main result in Theorem ~\ref{thm:PG}. 
\begin{lemma}\label{lem:dsepFuture}
{\it 
Consider a generative graph, $G=(V,A),$ consisting of $N$ nodes. Let
$Z=\{v_1,\dots , v_n\}\subset V$ be the set of $n$ perturbed nodes
where each perturbation is described by
\eqref{eq:corruptionModel}. Denote the data-streams as follows: $U_Z:=\{u_i\}_{i\in Z}$ and $Y_{\bar Z}:=\{y_j\}_{j\in \bar Z}$ where $\bar Z=V\setminus Z$. Let the measured data-streams be $\mathcal{W} =U_Z\cup Y_{\bar Z}=\{w _1, w _2, \dots , w _N\}$. Let $G'
=(V',A')$ be the dynamic Bayesian network (DBN) associated with $G$
and $G'_Z =(V'_Z,A'_Z)$ be the perturbed DBN. If $i\to j \notin A$ and
if a trail in $G'_Z$ between $w _i^{(t-1)}$ and $w _j[t]$ contains a
node $\alpha _{b_m}[t_m]$ such that $t_m\geq t$ and $b_m\in V$, then
for all $t>0$, the trail is not active given 
$\{w _j^{(t-1)},\mathcal{W} _{\bar j\bar i}^{(t-1)} \}$. 
}
\end{lemma}
\begin{proof}
Consider any trail from a node in $w_i^{(t-1)}$ to $w_j[t]$ in $G'_Z$. Denote this by $trlG'_Z:= w_i[t_1] = \alpha_{b_1}[t_1] -
\alpha_{b_2}[t_2] - \cdots - \alpha_{b_{r-1}}[t_{r-1}] -
\alpha_{b_r}[t_r] = w_j[t]$ where $0\leq t_1<t$. Here, $b_k$ denotes
the corresponding vertex in $V$ for $k=\{1,2,\dots , r\}$. Also,
$\alpha _{b_k}[t_k]=u_{b_k}[t_k]$ if $b_k \in Z$ or $\alpha
_{b_k}[t_k]=y_{b_k}[t_k]$ otherwise. For compact notation, set $\theta:=\{w _j^{(t-1)},\mathcal{W} _{\bar j\bar i}^{(t-1)} \}$.

{\it The trail has length at least 3.}
As $i\to j \notin A$ and if $j\notin Z$, then $y_j[t]$ does not
dynamically depend on process $y_i$ and clearly not on $u_i$. If $j\in
Z$, then by \eqref{eq:corruptionModel}, $u_j[t]$ does not dynamically
depend on $y_i$ nor $u_i$. Thus, there is no direct link of the form
$\alpha _i[t']\to \alpha _j[t"]$ in $G'_Z$, for any $t',t"$. In
particular, $w_i[t_1]\to w_j[t]\notin G'_Z$. Thus, there are at least
3 nodes in the trail, $trlG'_Z$. 

{\it Unobserved collider in trail.} 
Without loss of generality, choose $t_m = \max\{t_1,\ldots,t_{r-1} \} \ge t$. 
Consider the sub-trail $subtrl' := \alpha_{b_{m-1}}[t_{m-1}] - \alpha_{b_m}[t_m]
- \alpha_{b_{m+1}}[t_{m+1}]$ of $trlG'_Z$. By maximality of $t_m$,
$t_m\ge t_{m-1}$ and $t_m\ge t_{m+1}$. We will show that one of
$\alpha_{b_{m-1}}[t_{m-1}],\ \alpha_{b_m}[t_m],$ and $
\alpha_{b_{m+1}}[t_{m+1}]$ is a collider not in $\theta $ and therefore the
trail $trlG'_Z$ cannot be active given $\theta$. 

Suppose $t_m> t_{m-1}$ and $t_m> t_{m+1}$. Then, $subtrl'$ is of the form, $\alpha_{b_{m-1}}[t_{m-1}] \to \alpha_{b_m}[t_m] \leftarrow \alpha_{b_{m+1}}[t_{m+1}]$. Note that, as $t_m\geq t,$ it follows that neither  $\alpha_{b_m}[t_m]$ nor any of its descendants can be in $\theta$ and hence not observed.

Now, consider $t_m> t_{m-1}$ and $t_m= t_{m+1}$. (The case of $t_m> t_{m+1}$ and $t_m= t_{m-1}$ can be proven similarly). By the generative
model in \eqref{eq:genmodel}, by strict causality, for any node $p\in V$, $y_p[t_p]$ does
not dynamically depend on any $y_q[t_p]$ for $q \in \{p,
\mathcal{P}'(p)\}$. By the perturbation model described by
\eqref{eq:corruptionModel}, for any $q\in Z$, $u_q[t_q]$ dynamically
depends only on $\{u_q^{(t_q-1)},y_q^{(t_q)}\}$. As $t_m=t_{m+1}$, we
therefore have $b_m=b_{m+1}$ such that $b_m \in Z$ and, one of $\alpha
_{b_m}[t_m]$ and $\alpha _{b_{m+1}}[t_{m+1}]$ is actually a perturbed
measurement $u_{b_m}[t_m]$ while the other being $y_{b_m}[t_m]$. 

Suppose $\alpha _{b_m}[t_m]=u_{b_m}[t_m]$. Then, $\alpha
_{b_{m+1}}[t_{m+1}]=y_{b_m}[t_m]$. As $t_m>t_{m-1}$, $subtrl'$ is in
fact $\alpha_{b_{m-1}}[t_{m-1}] \to \alpha_{b_m}[t_m]=u_{b_m}[t_m]
\gets \alpha_{b_{m+1}}[t_{m+1}]=y_{b_m}[t_m]$. Therefore,
$\alpha_{b_m}[t_m]$ is a collider and as $t_m\ge t$, this node is not
observed in $\theta$.

Suppose instead that $\alpha _{b_m}[t_m]=y_{b_m}[t_m]$. Then,
$\alpha _{b_{m+1}}[t_{m+1}]=u_{b_m}[t_m]$. As $b_m\in Z$ and
maximality of $t_m$ implies $\alpha _{b_{m+2}}[t_{m+2}]\in
\{u_{b_m}^{(t_m-1)},y_{b_m}^{(t_m-1)}\}$. Thus, we have
$\alpha_{b_{m-1}}[t_{m-1}] -\alpha_{b_m}[t_m]=y_{b_m}[t_m] \to
\alpha_{b_{m+1}}[t_{m+1}]=u_{b_m}[t_m]\gets
\alpha_{b_{m+2}}[t_{m+2}]$ in $trlG'_Z$. Therefore,
$\alpha_{b_{m+1}}[t_{m+1}]$ is a collider not observed in $\theta$.
%
%
\end{proof}
The following theorem states that the perturbed graph precisely characterizes the spurious links which arise from probabilistic relationships that are spuriously introduced due to corruption.
\begin{theorem}\label{thm:PG}
{\it
Consider a generative graph, $G=(V,A),$ consisting of $N$ nodes. Let
$Z=\{v_1,\dots , v_n\}\subset V$ be the set of $n$ perturbed nodes
where each perturbation is described by
\eqref{eq:corruptionModel}. Denote the data-streams as follows: $U_Z:=\{u_i\}_{i\in Z}$ and $Y_{\bar Z}:=\{y_j\}_{j\in \bar Z}$ where $\bar Z=V\setminus Z$. Let the measured data-streams be $\mathcal{W} =U_Z\cup Y_{\bar Z}=\{w _1, w _2, \dots , w _N\}$. Let the perturbed graph be $G_Z=(V,A_Z)$ and its associated perturbed DBN be $G'_Z =(V'_Z,A'_Z)$.
If $i\to j\notin A_Z$, then d-sep($w_j[t], w_i^{(t-1)}\mid \{w_i^{(t-1)},\mathcal{W}_{\bar{j}\bar{i}}^{(t-1)}\})$ holds in $G'_Z$ for all $t>0$.
}
\end{theorem}
\begin{proof}
We will show that if $i\to j\notin A_Z$, then there is no trail between
$w_i^{(t-1)}$ and $w_j[t]$ that is active given 
$\{w_j^{(t-1)},\mathcal{W}_{\bar{j}\bar{i}}^{(t-1)}\}$ in $G'_Z,$ for all $t>0$. For rest of the proof, denote $\theta:=\{w _j^{(t-1)},\mathcal{W} _{\bar j\bar i}^{(t-1)} \}$.  
Note that if $i\to j\notin A_Z$, then there is no directed edge from
$i$ to $j$ in $G$, and every trail from $i$ to $j$ in $G$ violates at
least one of the conditions of Definition~\ref{def:pg}. We will consider these cases separately and show that no active
trail exists in $G'_Z$ in each case.
Denote a trail connecting a node in $w_i^{(t-1)}$
and $w_j[t]$ in $G'_Z$, by $trlG'_Z:= w_i[t_1] = \alpha_{b_1}[t_1] -
\alpha_{b_2}[t_2] - \cdots - \alpha_{b_{r-1}}[t_{r-1}] -
\alpha_{b_r}[t_r] = w_j[t]$ where $0\leq t_1<t$ and $b_k$ denotes
the corresponding vertex in $V$ for $k=\{1,2,\dots , r\}$. Here, $\alpha _v[t_v]=u_v[t_v]$ if $v\in Z$ or $\alpha _v[t_v]=y_v [t_v]$ otherwise. Using Lemma \ref{lem:dsepFuture}, if any $t'$ in $\{t_2, \dots , t_{r-1}\}$ is such that $t'\ge t$, then $trlG'_Z$ is not active. Now, consider $0\le t_1,t_2,t_3,\dots , t_{r-1}<t$. Construct a trail in $G$, $trlG:=i=v_1-v_2-v_3\dots v_{k-1}-v_k=j$ from the trail $trlG'_Z$: $w_i[t_1]=\alpha _{b_1}[t_1]-\alpha _{b_2}[t_2]- \dots - \alpha _{b_{r-1}}[t_{r-1}]-\alpha _{b_r}[t_r]=w_j[t]$ as follows: 
\begin{algorithm}
	\begin{algorithmic}
		\Statex \textbf{Initialize:} $k=1$ and $v_1=b_1$.
		\For{$l=1:r-1 $}
			\If{$b_{l+1} \ne b_l$ in $\alpha _{b_{l}}[t_l]-\alpha _{b_{l+1}}[t_{l+1}]$ along $trlG'_Z$} 
				\State Set $v_{k+1}= b_{l+1}$.
				\State Add edge $v_k-v_{k+1}$ with the
                                same direction as \State $\alpha
                                _{b_l}[t_l]-\alpha
                                _{b_{l+1}}[t_{l+1}]$.
                                \State Set $s_k = t_l$ and $\tau_{k+1}
                                = t_{l+1}$
				\State Set $k=k+1$
			\EndIf
		\EndFor
	\end{algorithmic}
\end{algorithm}

\noindent Additionally, note that $v_k - v_{k+1}$ corresponds to an edge
$\alpha_{v_k}[s_k] - \alpha_{v_{k+1}}[\tau_{k+1}]$ in $G'_Z$.

 Now, let us reason out why such a construction is always
feasible. To this, we claim that for any successive pair $\alpha
_{b_{l}}[t_l]-\alpha _{b_{l+1}}[t_{l+1}]$, either $b_l=b_{l+1}$ or,
$b_l\ne b_{l+1}$ and $b_{l}-b_{l+1} \in A$ with the same direction as
in $\alpha _{b_{l}}[t_l]-\alpha _{b_{l+1}}[t_{l+1}]$. Assume $\alpha
_{b_{l}}[t_l]\to \alpha _{b_{l+1}}[t_{l+1}]$. (The case of $\alpha
_{b_{l}}[t_l]\leftarrow \alpha _{b_{l+1}}[t_{l+1}]$ is similar). Then,
either $t_l=t_{l+1}$ or $t_l<t_{l+1}$. Consider, $t_l=t_{l+1}$. Then,
the link must have the form $y_{b_l}[t_l] \to u_{b_l}[t_l]$, as this
is the only instantaneous influence defined in (\ref{eq:genmodel}) or
(\ref{eq:corruptionModel}). Thus, $b_l = b_{l+1}$ in this case.

Suppose, $t_l< t_{l+1}$. Either, $b_{l+1}\in Z$ or $b_{l+1}\notin Z$. Consider $b_{l+1}\in Z$. By the perturbation model described by \eqref{eq:corruptionModel}, $\alpha _{b_l}[t_l]\in \{y_{b_{l+1}}^{(t_{l+1}-1)}, u_{b_{l+1}}^{(t_{l+1}-1)}\}$. Therefore, $b_l=b_{l+1}$. Suppose, $b_{l+1}\notin Z$. Then, $\alpha _{b_{l+1}}[t_{l+1}]=y_{b_{l+1}}[t_{l+1}]$. By the generative model in \eqref{eq:genmodel}, we either have dynamic dependence on self-history or history of other nodes. That is, $\alpha _{b_l}[t_l]\in \{ y_{b_l}^{(t_{l+1}-1)},\underset{q\in \mathcal{P}'(b_{l+1})}{\cup}y_q^{(t_{l+1}-1)}\}$. Then, $b_l=b_{l+1}$ when there is dependence on self-history. Otherwise, $b_{l}\in \mathcal{P}'(b_{l+1})$. Thus, $b_{l}\to b_{l+1}\in A$. Let us consider an example- from a trail of the form $u_1[t_1] \leftarrow y_1[t_2] \leftarrow y_{2}[t_3] \rightarrow
y_3[t_4] \rightarrow y_3[t_5] \rightarrow u_3[t]$ in $G'_Z$, a trail
$trlG$ in $G$ can be constructed as $1 \leftarrow 2\rightarrow 3$. 

Additionally, we may assume that for $m=2,\cdots , r-1$ we have that $\alpha _{b_m}[t_m]\ne w_i[t_m]$ in $trlG'_Z$. If $\alpha _{b_m}[t_m]= w_i[t_m]$ for some $m>1$, then the sub-trail of $trlG'_Z$, $w_i[t_m]=\alpha _{b_m}[t_m]-\alpha _{b_{m+1}}[t_{m+1}]-\cdots -\alpha _r[t_r]=w_j[t]$ is a trail from $w_i[t_m]\in w_i^{(t-1)}$ to $w_j[t]$. This trail is of strictly shorter length than $trlG'_Z$. Thus, if the shorter trail cannot be active then the longer trail, $trlG'_Z$, cannot be active either. Also, by following the construction procedure described above, this condition implies that $v_l\ne i$ for $l=2,3,\cdots , k$ in $trlG$. Call this condition $loop_i$. 
To summarize, let $trlG:=i=v_1-v_2-v_3\dots v_{k-1}-v_k=j$ be the trail in $G$ constructed by following the above procedure from the trail $trlG'_Z$: $w_i[t_1]=\alpha _{b_1}[t_1]-\alpha _{b_2}[t_2]- \dots - \alpha _{b_{r-1}}[t_{r-1}]-\alpha _{b_r}[t_r]=w_j[t]$. Since, $i\to j \notin A_Z$, this trail must violate any of the conditions P\ref{case:end}), P\ref{case:collider}) and P\ref{case:notCollider}). We will now consider these cases separately and prove that there is no corresponding active trail in $G'_Z$.

If condition P\ref{case:end}) is violated, then $trlG$ must have that
$j\notin Z$ and $v_{k-1}\leftarrow j$. In this case, $w_j =
y_j$. Then, either $b_{r-1}=j$ or $b_{r-1}\ne j$. By construction of
$trlG$, if $b_{r-1}\ne j$, then $b_{r-1}=v_{k-1}$. As
$v_{k-1}\leftarrow j$, we must then have $\alpha
_{b_{r-1}}[t_{r-1}]\gets  \alpha _{b_r}[t_r]$. However, this implies
$t_r=t < t_{r-1}$ which violates the condition that $0\le
t_1,t_2,t_3,\dots , t_{r-1}<t$. Thus, $b_{r-1}=j$. That is, $\alpha
_{b_{r-1}}[t_{r-1}]=y_{j}[t_{r-1}]$. As $t_{r-1}<t$ and $j\notin Z$ we
have $\alpha _{b_{r-1}}[t_{r-1}]=y _{j}[t_{r-1}]\to \alpha
_{b_{r}}[t]=y_{j}[t]$ as a sub-trail of $trlG'_Z$. Clearly, $y
_{j}[t_{r-1}]$ is not a collider. As $t_{r-1}<t$, we have $ y
_{j}[t_{r-1}] \in \theta $. Thus the trail cannot be active. 

Recall the definitions of $s_k$ and $\tau_{k+1}$ during construction of
the trail in $G$. 
If condition P\ref{case:collider}) is violated, then a sub-path of $trlG$, $v_{m-1} \to v_m \leftarrow v_{m+1}$, must have a
collider, $v_m$, such that $v_m \notin Z$
and $v_{m+1}\notin Z $ where $m=\{2,3\cdots ,k-1\}$. If $v_{m+1} = j$
and $\tau_{m+1}=t$,
P\ref{case:end}) also fails, and the argument above shows that the
trail in $G'_Z$ is not active. If $v_{m+1} = j$ and $\tau_{m+1} <t$ then 
we have that $\alpha _{v_{m+1}}[\tau_{m+1}]=y_{v_{m+1}}[\tau_{m+1}] \in \theta$ which is an observed
node along the trail and is not a collider. Thus, the trail $trlG'_Z$
cannot be active. So, assume that $v_{m+1} \ne j$. By condition
$loop_i$, $m+1\ne i$. As $v_m \leftarrow v_{m+1}\in trlG$, by
construction we must have $y_{v_m}[s_m]=\alpha _{v_m}[s_m] \gets
\alpha _{v_{m+1}}[\tau_{m+1}]=y_{v_{m+1}}[\tau_{m+1}] $ along $trlG'_Z$
with $\tau_{m+1}<s_m<t$. Note that since $v_{m+1}\notin Z$ and $\tau_{m+1}<t$, $\alpha _{v_{m+1}}[\tau_{m+1}] = y_{v_{m+1}}[\tau_{m+1}]$ is an
observed non-collider in $\theta$. Thus, the trail cannot be active. 

Finally consider the case that P\ref{case:notCollider}) is
violated. Then along the trail, $trlG$, in $G$, there must be a sub-trail $v_{m-1} - v_m - v_{m+1}$ such that the intermediate
node, $v_m$, is not a collider and $v_m \notin Z$. As $v_m$ is not a collider, there is one outgoing directed edge from $v_m$ in the trail $trlG$ to either $v_{m-1}$ or $v_{m+1}$. By construction, there must be a corresponding node $\alpha_{v_m}[t_f]$ in
the trail $trlG'_Z$ such that it has an outgoing edge to either $\alpha _{v_{m-1}}[t_{p}]$ or $\alpha _{v_{m+1}}[t_{q}]$ for some $t_p>t_m$ or $t_q > t_m$ respectively. Clearly, there is one $\alpha _{v_m}[t_m]$ in $trlG'_Z$ which is a non-collider. Then, as $v_m\notin Z$, we
must have that $\alpha_{v_m}[t_m] = w_{v_m}[t_m] =
y_{v_m}[t_m]$. Note that $v_m \ne i$ by condition $loop_i$. As $t_m < t$, $\alpha _{v_m}[t_m]$ is an intermediate non-collider node in $\theta $ and is thus observed. Hence, $trlG'_Z$ cannot be active. 
\end{proof}
We will now show that if conditional directed information, $I(w_i \to w_j \parallel \mathcal{W}_{\bar{j}\bar{i}})$, are computed using corrupt data-streams, and were applied for causal structure inference, then spurious links in the graph would result.
\begin{corollary}\label{cor:DItoPG}
\it
Consider a generative graph, $G=(V,A),$ consisting of $N$ nodes. Let
$Z=\{v_1,\dots , v_n\}\subset V$ be the set of $n$ perturbed nodes
where each perturbation is described by
\eqref{eq:corruptionModel}. Denote the data-streams as follows: $U_Z:=\{u_i\}_{i\in Z}$ and $Y_{\bar Z}:=\{y_j\}_{j\in \bar Z}$ where $\bar Z=V\setminus Z$. Let the measured data-streams be $\mathcal{W} =U_Z\cup Y_{\bar Z}=\{w _1, w _2, \dots , w _N\}$. Let the perturbed graph be $G_Z=(V,A_Z)$.
If $I(w_i \to w_j \parallel
\mathcal{W}_{\bar{j}\bar{i}})>0$, then $i\to j\in A_Z$.
\end{corollary}
\begin{proof}
We will show that if $i\to j \notin A_Z$, then $I(w_i \to w_j \parallel
\mathcal{W}_{\bar{j}\bar{i}})=0$. Suppose, $i\to j\notin A_Z$. Let $G'_Z
=(V',A'_Z)$ be the perturbed dynamic Bayesian network (DBN) associated with the perturbed graph, $G_Z$. Then, using Theorem ~\ref{thm:PG}, for all $t>0$, d-sep($w_j[t], w_i^{(t-1)}\mid w_i^{(t-1)},\mathcal{W}_{\bar{j}\bar{i}}^{(t-1)})$ holds in $G'_Z$. In other words, this implies $P_{w_j[t]\mid w_j^{(t-1)},w_i^{(t-1)},\mathcal{W}_{\bar{j}\bar{i}}^{(t-1)}}=P_{w_j[t]\mid w_j^{(t-1)},\mathcal{W}_{\bar{j}\bar{i}}^{(t-1)}}$ will hold true for all $t$ and thus, $I(w_i \to w_j \parallel
\mathcal{W}_{\bar{j}\bar{i}})=0$. 
\end{proof}

The following example illustrates the intuition behind the presence of active trails and hence, the spurious links in the perturbed graph. 
\begin{example}\label{eg:intuition}
Consider a generative graph as shown in Figure ~\ref{fig:proof_idealG}). Suppose node $3$ is subject to data-corruption and let $u_3$ be its measured data-stream. Denote the measured data-streams at nodes 1 and 2 as $y_1$ and $y_2$. $u_3$ is related to its ideal counterpart
$y_3$ via \eqref{eq:corruptionModel}. The measured data streams are $\{w_1=y_1,w_2=y_2,w_3=u_3\}$.
 The perturbed graph $G_Z$, is constructed as defined in definition ~\ref{def:pg} and is shown in figure ~\ref{fig:Gzproof}). The corresponding perturbed DBN, $G'_Z$, is shown for 3 time steps in figure ~\ref{fig:pdbnProof}). We will reason out the presence and absence of an edge in $G_Z$ by identifying the presence and absence of active trails in the perturbed DBN.

Consider $1\to 3 \in A_Z$. There is a trail $w_1[0]=y_1[0]\to w_2[1]=y_2[1]\gets y_3[0]\to y_3[1]\to w_3[2] \in A'_Z$. Note that the collider $w_2[1]$ is observed. Therefore, the trail is active given $\{w_3^{(1)},w_2^{(1)}\}$.

Take the edge $2\to 3 \in A_Z$. There exists a trail $w_2[1]=w_2[1]\gets y_3[0]\to w_3[2]$ in $G'_Z$. Note that the node $y_3[0]$ is not a collider and is not observed. Thus, the trail is active given $\{w_3^{(1)},w_1^{(1)}\}$. 

The edges $3\to1$ and $2\to 1$ are absent in $G_Z$. Ideally, we look for a trail from $w_3^{(t-1)}$ to $w_1[t]=y_1[t]$ that is active given $\{w_1^{(t-1)},w_2^{(t-1)}\}$ and a trail from $w_2^{(t-1)}$ to $w_1[t]=y_1[t]$ that is active given $\{w_1^{(t-1)},w_3^{(t-1)}\}$. Note that every trail from $w_3^{(t-1)}$ and $w_2^{(t-1)}$ to $w_1[t]$ traverses through a node in $w_1^{(t-1)}$ which is in the observed set and this holds for all $t$. This blocks the information flow along the trail. Therefore, all these trails are inactive. 
\end{example}
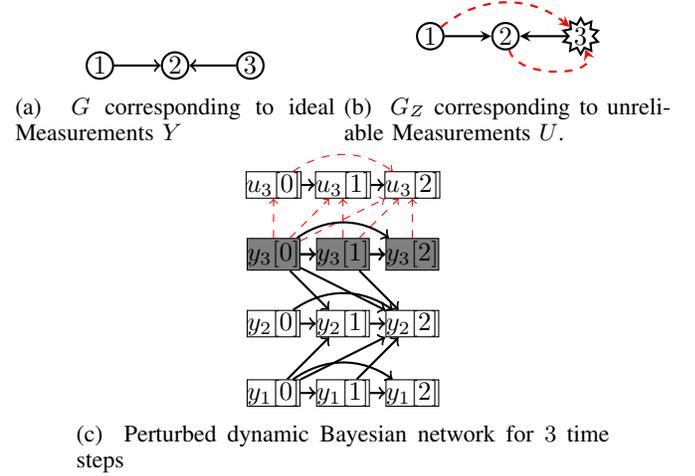
\begin{figure}[t]
  \centering
  \vspace{2mm} 
  \subcaptionbox{\label{fig:proof_idealG} $G$ corresponding to ideal Measurements $Y$}{\begin{subfigure}{0.48\columnwidth}
    \centering
        \begin{tikzpicture}[scale=0.25]
                \tikzstyle{vertex}=[circle,fill=none,minimum size=10pt,inner sep=0pt,thick,draw]
        \tikzstyle{pvertex}=[star,star points=10,fill=white,minimum size=10pt,inner sep=0pt,thick,draw]
        
          \node[vertex] (n1) {$1$};
          \node[vertex, right of=n1] (n2) {$2$};
          \node[vertex,right of=n2] (n3) {$3$};
                    
          \draw[->,thick] (n1)--(n2);
          \draw[->,thick] (n3)--(n2);    
        \end{tikzpicture}
        
  \end{subfigure}}
    \subcaptionbox{
        \label{fig:Gzproof} $G_Z$ corresponding to unreliable Measurements $U$.
        }{\begin{subfigure}{0.49\columnwidth}
    \centering
        \begin{tikzpicture}[scale=0.25]
                \tikzstyle{vertex}=[circle,fill=none,minimum size=10pt,inner sep=0pt,thick,draw]
        \tikzstyle{pvertex}=[star,star points=10,fill=white,minimum size=10pt,inner sep=0pt,thick,draw]
          \node[vertex] (n1) {$1$};
          \node[vertex, right of=n1] (n2) {$2$};
          \node[pvertex,right of=n2] (n3) {$3$};

          \draw[-{stealth[scale=3.0]},thick] (n1)--(n2);
          \draw[-{stealth[scale=3.0]},thick] (n3)--(n2);
          \draw[-{stealth[scale=3.0]},red,thick,dashed] (n2)to[out=285,in=290](n3);
          \draw[-{stealth[scale=3.0]},red,thick,dashed] (n1) to[out=40,in=140] (n3);
        \end{tikzpicture}
          \end{subfigure}}
          \subcaptionbox{\label{fig:pdbnProof}
  Perturbed dynamic Bayesian network for $3$ time steps}{\begin{subfigure}{0.8\columnwidth}
   \centering
 \begin{tikzpicture}[scale=0.7]
 \tikzstyle{vertex}=[fill=none,minimum size=10pt,inner sep=0pt,thin,draw]
  \tikzstyle{Zuvertex}=[fill=gray,minimum size=12pt,inner sep=0pt,thin,draw]
  \node[Zuvertex] (n1) {$y_3[0]$};
  \node[vertex] (n16) at ($(n1)-(0,-3.75em)$) {$u_3[0]$};
    \node[vertex] (n17) at ($(n16)-(-3.75em,0)$) {$u_3[1]$};
      \node[vertex] (n18) at ($(n17)-(-3.75em,0)$) {$u_3[2]$}; 
  \node[vertex] (n2) at ($(n1)-(0,3.75em)$) {$y_2[0]$};%

        \node[vertex] (n3) at ($(n2)-(0,3.75em)$) {$y_1[0]$};%
	  \node[Zuvertex] (n6) at ($(n1)-(-3.75em,0)$) {$y_3[1]$};

  \node[vertex] (n7) at ($(n6)-(0,3.75em)$) {$y_2[1]$};%

        \node[vertex] (n8) at ($(n7)-(0,3.75em)$) {$y_1[1]$};%
	\node[Zuvertex] (n11) at ($(n6)-(-3.75em,0)$) {$y_3[2]$};

  \node[vertex] (n12) at ($(n11)-(0,3.75em)$) {$y_2[2]$};%

        \node[vertex] (n13) at ($(n12)-(0,3.75em)$) {$y_1[2]$};%

  \draw[->,thick] (n1)--(n6);
  \draw[->,thick] (n2)--(n7);
  \draw[->,thick] (n3)--(n8);
  
  \draw[->,red,dashed] (n1)--(n16);  
    \draw[->,red,dashed] (n6)--(n17);
      \draw[->,red,dashed] (n11)--(n18);
         \draw[->,red,dashed] (n1)--(n17);
               \draw[->,red,dashed] (n1)--(n18);
                     \draw[->,red,dashed] (n6)--(n18);   
        \draw[->,thick] (n16)--(n17);
          \draw[->,thick] (n17)--(n18);
  
  \draw[->,thick] (n1)--(n6);
  \draw[->,thick] (n1)--(n7);

  \draw[->,thick] (n6)--(n11);
  \draw[->,thick] (n7)--(n12);
  \draw[->,thick] (n8)--(n13);
  
  \draw[->,thick] (n6)--(n11);
  \draw[->,thick] (n6)--(n12);
  
 \draw[->,thick] (n3)--(n12);
  \draw[->,thick] (n8)--(n12);  
   \draw[->,thick] (n3)--(n7);
  \draw[->,thick] (n1) to[out=35,in=145] (n11);
    \draw[->,red,dashed] (n16) to[out=35,in=145] (n18);
  \draw[->,thick] (n1)--(n12);
  
  \draw[->,thick] (n2) to[out=35,in=145] (n12);
  \draw[->,thick] (n3) to[out=35,in=145] (n13); 
  
\end{tikzpicture}
\end{subfigure}}
        \caption{
    \label{fig:proofEg} This figure illustrates the proof of Theorem ~\ref{thm:PG}
  }
\end{figure}
%
\begin{remark}\label{rem:suff}
The results in Theorem \ref{thm:PG} and Corollary \ref{cor:DItoPG} respectively shows that existence of active trails is the PDBN and non-zero conditional directed information is sufficient to infer the presence of a directed link in the perturbed graph. However, under a mild assumption on the generative and the perturbation model, it can be shown that the respective conditions are also necessary to detect a directed link in the perturbed graph.
\end{remark}
\begin{assumption}\label{assume:relax}
Let the following conditions on the generative and the perturbation model hold:
\begin{enumerate}[C1)]
\item \label{c1}In the generative model ~\eqref{eq:genmodel}, for all agents 
  $i \in \{1,2,\dots , N\}$, and all $j\in \pa'(i)$, there is a number
  $k_{ij}\ge 1$ such that $y_j[t-k_{ij}]$ is an argument of $f_i$.
\item \label{c2}For all perturbed nodes $i\in Z$, in the perturbation model ~\eqref{eq:corruptionModel}, there is a number $k_i\ge 1$ such that $g_i$ always takes $y_i[t-k_i]$ as it's argument.
\end{enumerate}
In addition, let at least one of the following conditions on corruption model hold:
\begin{enumerate}[B1)]
\item \label{c3} If a node $i\in Z$, then there is a number $k'_i\geq 1$ such that $y_i[t-k'_i]$ is an argument of $f_i$ in ~\eqref{eq:genmodel}. 
\item \label{c4} If a node $i\in Z$, then $y_i[t]$ is an argument of $g_i$ in ~\eqref{eq:corruptionModel}. 
\end{enumerate}
\end{assumption}
The following theorem asserts that if $i\to j\in A_Z$ then there exists a corresponding active trail in perturbed DBN.
\begin{theorem}\label{thm:PGrelaxed}
{\it
Consider a generative graph, $G=(V,A),$ consisting of $N$ nodes. Let
$Z=\{v_1,\dots , v_n\}\subset V$ be the set of $n$ perturbed nodes
where each perturbation is described by
\eqref{eq:corruptionModel}. Denote the data-streams as follows: $U_Z:=\{u_i\}_{i\in Z}$ and $Y_{\bar Z}:=\{y_j\}_{j\in \bar Z}$ where $\bar Z=V\setminus Z$. Let the measured data-streams be $\mathcal{W} =U_Z\cup Y_{\bar Z}=\{w _1, w _2, \dots , w _N\}$. Suppose, the generative model and the perturbation model satisfies the conditions for dynamics that is mentioned in Assumption ~\ref{assume:relax}. If there is a directed edge from $i$ to $j$ in perturbed graph, $G_Z=(V,A_Z)$, then there exists a trail between a node in $w_i^{(t-1)}$ and $w_j[t]$ that is active given 
$\{w_j^{(t-1)},\mathcal{W}_{\bar{j}\bar{i}}^{(t-1)}\}$ in $G'_Z,$ for some $t>0$.
}
\end{theorem}
\begin{proof}
The proof is given in appendix ~\ref{appx:relax}.
\end{proof}
Under the following assumption we can in fact show that $I(w_i \to w_j \parallel
w_{\bar{j}\bar{i}})>0$ is also a necessary condition for $i\to j\in A_Z$ as showin in Corollary ~\ref{cor:PGtoDI}. 
\begin{assumption}\label{assume:faithfulBN}We assume that the generative model in ~\eqref{eq:genmodel} and the perturbation model in ~\eqref{eq:corruptionModel} are such that the corresponding DBN and PDBN are faithful Bayesian networks. Moreover, we consider positive joint distributions for the random processes $Y$ and $U$. 
\end{assumption}
\begin{corollary}\label{cor:PGtoDI}
\it
Under assumption \ref{assume:faithfulBN} and dynamics as described in Assumption \ref{assume:relax}, if $i\to j\in A_Z$, then $I(w_i \to w_j \parallel
w_{\bar{j}\bar{i}})>0$.
\end{corollary}
\begin{proof}
By theorem \ref{thm:PGrelaxed}, if $i\to j \in A_Z$, then there exists an trail in PDBN  between $w_i^{(t-1)}$ and $w_j[t]$ that is active given 
$\{w_j^{(t-1)},\mathcal{W}_{\bar{j}\bar{i}}^{(t-1)}\}$ in $G'_Z,$ for some $t>0$. Under faithfulness assumption, this implies $P_{w_j[t]\mid w_j^{(t-1)},w_i^{(t-1)},\mathcal{W}_{\bar{j}\bar{i}}^{(t-1)}}\ne P_{w_j[t]\mid w_j^{(t-1)},\mathcal{W}_{\bar{j}\bar{i}}^{(t-1)}}$. Thus, $I(w_i \to w_j \parallel
\mathcal{W}_{\bar{j}\bar{i}})>0$. 
\end{proof}
\section{Estimation of Directed information}\label{sec:DIest}
In \cite{jiao2013universal}, consistency results for estimating directed information(DI) between a pair of random processes from data was proposed. However, in this article we extend the methods to determine the directed information between two processes conditioned on a set of other random processes. We provide consistency results of the estimator by showing convergence in almost sure sense(denoted as P-a.s). 
\subsection{Pairwise Estimation of Directed Information}\label{subsec:reviewJiao}
Here, we present the definition used for directed information estimator proposed in \cite{jiao2013universal}  Before that, the following notion of universal probability assignment is needed. 
\subsubsection{Universal Probability Assignment}\label{subsec:defineQ}
Let $Q$ denote a sequential probability assignment for a sequence $x$. That is, the conditional probability mass function(pmf) for $x[i]$ given $x^{(i-1)}$ is given by $Q(x[i]\mid x^{(i-1)})$. The joint pmf for $x^{(n)}$ is given by $Q(x^{(n)})=Q(x[0])Q(x[1]\mid x[0])Q(x[2]\mid x^{(1)}])\cdots Q(x[n]\mid x^{(n-1)})$. 
\begin{definition}[Universal Probability Assignment]
Let $P$ be the true joint pmf for $x^{(n)}$. Then, a probability assignment $Q$ is called as \emph{universal} if the following holds:
\begin{equation}
\lim _{n\to \infty}\frac{1}{n}\mean\left[ \log \frac{P(x^{(n)})}{Q(x^{(n)})}\right]=0.
\end{equation}
\end{definition}
Context tree weighting(CTW) algorithm developed by \cite{willems1995theCTW} will be used for computing sequential probability assignment. 
\subsubsection{DI estimation}\label{subsec:DIest}
Let $X$ and $Y$ be jointly stationary and  ergodic processes. The directed information from $X$ to $Y$ can be expressed in terms of the entropy as follows:
\begin{equation}
I(X\to Y)=H(Y)-H(Y\parallel X)
\end{equation}
where $H(Y)=\mean [-\log P(Y)]$ and $H(Y\parallel X)=\mean [-\log P(Y\parallel X)]$ denotes the entropy of $Y$ and the causally conditioned entropy \cite{kramer1998thesis} respectively. 

The directed information rate (DIR) from $X$ to $Y$ is defined as:
\begin{equation}
I_r(X\to Y)=\lim _{n\to \infty}\frac{1}{n}I(X^{(n)}\to Y^{(n)}).
\end{equation}
Let $H_r(Y):=\lim _{n\to \infty}\frac{1}{n}H(Y^{(n)})$ and let $H_r(Y\parallel X):=\lim _{n\to \infty}\frac{1}{n}H(Y^{(n)}\parallel X^{(n)})$. Thus, if $H_r(Y)$ and $H_r(Y\parallel X)$ converge, then $I_r$ is convergent. That is,
\begin{equation}
I_r=H_r(Y)-H_r(Y\parallel X).
\end{equation}

In \cite{jiao2013universal}, the following DIR estimator was defined:
\begin{multline}\label{eq:defineI}
\hat I(X^{(n)}\to Y^{(n)}) =\frac{1}{n}\sum _{i=1}^n \sum _{y[i]}Q(y[i]\mid X^{(i-1)},Y^{(i-1)})\cdot\\  \log \frac{1}{Q(y_i\mid Y^{(i-1)})}
\\  -\frac{1}{n}\sum _{i=1}^n\sum _{y_i}Q(y[i]\mid X^{(i-1)},Y^{(i-1)})\cdot \\ \log \frac{1}{Q(y[i]\mid X^{(i-1)},Y^{(i-1)})}
\end{multline}
We will extend the above to define conditional directed information as described in the following subsection.
\subsection{Estimation of Conditional Directed Information}\label{subsec:di_ext}
Let $X,Y,Z$ be jointly stationary and ergodic processes. The conditional directed information from $X$ to $Y$ conditioned on $Z$ can be expressed in terms of the entropy as follows:
\begin{equation}
I(X\to Y\parallel Z)=H(Y\parallel Z)-H(Y\parallel X,Z).
\end{equation}

The causally conditioned directed information rate (DIR) from $X$ to $Y$ now is defined as:
\begin{equation}
I_r(X\to Y\parallel Z)=\lim _{n\to \infty}\frac{1}{n}I(X^{(n)}\to Y^{(n)}\parallel Z^{(n)}).
\end{equation}
Let $H_r(Y\parallel X,Z):=\lim _{n\to \infty}\frac{1}{n}H(Y^{(n)}\parallel X^{(n)},\ Z^{(n)})$. Thus, if $H_r(Y\parallel Z)$ and $H_r(Y\parallel X,Z)$ converge, then $I_r$ is convergent. That is,
\begin{equation}
I_r=H_r(Y\parallel Z)-H_r(Y\parallel X,Z).
\end{equation}
The conditional directed information estimator $\hat I(X^{(n)}\to Y^{(n)}\parallel Z^{(n)})$ is defined as under:
\begin{multline}\label{eq:conditionalDI}
\hat I(X^{(n)}\to Y^{(n)}\parallel Z^{(n)}) =\\ \frac{1}{n}\sum _{i=1}^n\sum _{y[i]}Q(y[i]\mid X^{(i-1)},Y^{(i-1)},Z^{(i-1)})\cdot \\  \log \frac{1}{Q(y[i]\mid Y^{(i-1)},Z^{(i-1)})}
\\  -\frac{1}{n}\sum _{i=1}^n\sum _{y[i]}Q(y[i]\mid X^{(i-1)},Y^{(i-1)},Z^{(i-1)})\cdot \\ \log \frac{1}{Q(y[i]\mid X^{(i-1)},Y^{(i-1)},Z^{(i-1)})}
\end{multline}
The following theorem establishes the consistency result in estimating conditional DIR as defined in \eqref{eq:conditionalDI}. The proof is given in appendix ~\ref{appx:conditional}.
\begin{theorem}\label{thm:convergeConditionalI}
Let $Q$ be the probability assignment in the CTW algorithm. Suppose, $X,Y,Z$ are jointly stationary irreducible aperiodic finite-alphabet Markov processes whose order is bounded by the prescribed tree depth of the CTW algorithm. Then,
\begin{equation}\label{eq:condIconverge}
\lim _{n\to \infty} \hat I(X^{(n)}\to Y^{(n)}\parallel Z^{(n)})=I_r(X\to Y\parallel Z) \qquad \mbox{P-a.s},
\end{equation}
\end{theorem}

\subsection{Simulation Results}\label{sec:sim}
To verify the predictions of Theorem ~\ref{thm:PG}, we first performed a simulation on a network consisting of 3 nodes with a single node being perturbed and on a network consisting of 6 nodes, of which 2 are corrupt. We estimate the directed information rates (DIR), which are DI estimates that are averaged along the sequence length till the horizon. We used the estimator described in \eqref{eq:conditionalDI} to compute DIR. For both the networks, the horizon length is chosen as $10^4$. The DIR estimates were
then averaged over $50$ trials.
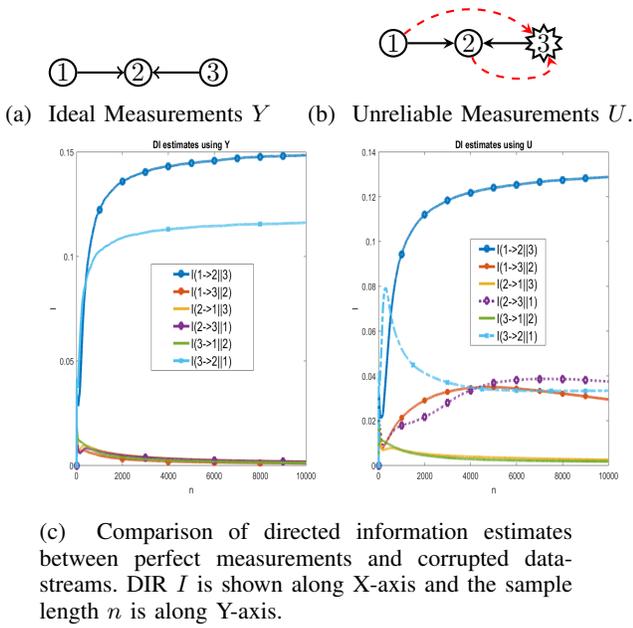
\begin{figure}[t]
  \centering
  \vspace{2mm} 
  \subcaptionbox{\label{fig:Col} Ideal Measurements $Y$}{\begin{subfigure}{0.48\columnwidth}
    \centering
        \begin{tikzpicture}[scale=0.25]
                \tikzstyle{vertex}=[circle,fill=none,minimum size=10pt,inner sep=0pt,thick,draw]
        \tikzstyle{pvertex}=[star,star points=10,fill=white,minimum size=10pt,inner sep=0pt,thick,draw]
          \node[vertex] (n1) {$1$};
          \node[vertex, right of=n1] (n2) {$2$};
          \node[vertex,right of=n2] (n3) {$3$};
                    
          \draw[->,thick] (n1)--(n2);
          \draw[->,thick] (n3)--(n2);    
        \end{tikzpicture}
        
  \end{subfigure}}
    \subcaptionbox{
        \label{fig:perturb 3} Unreliable Measurements $U$.
        }{\begin{subfigure}{0.49\columnwidth}
    \centering
        \begin{tikzpicture}[scale=0.25]
                \tikzstyle{vertex}=[circle,fill=none,minimum size=10pt,inner sep=0pt,thick,draw]
        \tikzstyle{pvertex}=[star,star points=10,fill=white,minimum size=10pt,inner sep=0pt,thick,draw]
          \node[vertex] (n1) {$1$};
          \node[vertex, right of=n1] (n2) {$2$};
          \node[pvertex,right of=n2] (n3) {$3$};

          \draw[-{stealth[scale=3.0]},thick] (n1)--(n2);
          \draw[-{stealth[scale=3.0]},thick] (n3)--(n2);
          \draw[-{stealth[scale=3.0]},red,thick,dashed] (n2)to[out=285,in=290](n3);
          \draw[-{stealth[scale=3.0]},red,thick,dashed] (n1) to[out=40,in=140] (n3);
        \end{tikzpicture}
          \end{subfigure}}
          \subcaptionbox{\label{fig:diCol}
  Comparison of directed information estimates between perfect measurements and corrupted data-streams. DIR $I$ is shown along X-axis and the sample length $n$ is along Y-axis.}{\begin{subfigure}{0.8\columnwidth}
   \includegraphics[height=5cm,width=7.8cm]{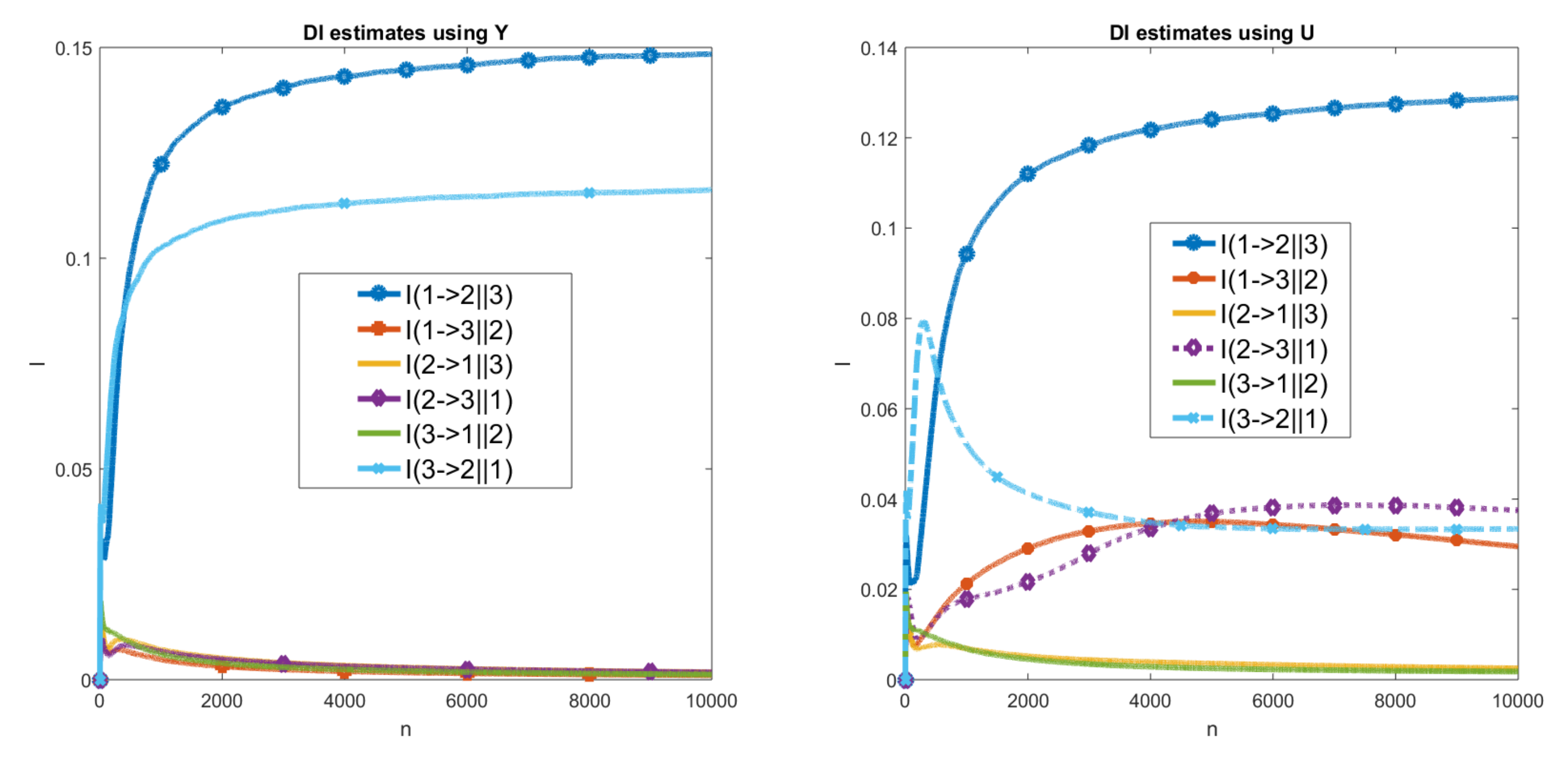}
    \end{subfigure}}
        \caption{
    \label{fig:singlePert} This figure shows how unreliable measurements at node 3 can result in spuriously inferring a direct dynamic influence of node $1$ on the third and a spurious influence of node $2$ on node $3$. 
  }
\end{figure}
\subsubsection{Single node Perturbation}
Consider a network consisting of 2 nodes with a common child as shown in Fig. ~\ref{fig:Col}). The dynamic interactions in the true generative model are as follows:
\begin{eqnarray*}
y_1[t]&=&e_1[t],\\
y_2[t]&=&y_1[t-1]+y_3[t-1]+e_2[t],\\
y_3[t]&=&e_3[t]
\end{eqnarray*} where $e_1[t] \sim $ Bernouilli(0.7), $e_2[t] \sim$ Bernouilli(0.4) and $e_3[t] \sim$ Bernouilli(0.6) and `$+$' is logical `OR' operation. Each of $y_1[t],y_2[t]$ and $y_3[t]$ has a finite alphabet $\{0,1\}.$

The perturbation considered here is the time-origin uncertainty at node 3. The corruption model takes the form:
\begin{equation*}
u_3[t] = \begin{cases}
  y_3[t-2], & \textrm{ with probability } 0.5 \\
  y_3[t], & \textrm{ with probability } 0.5.
\end{cases}
\end{equation*} 
The perturbed graph predicted by Theorem ~\ref{thm:PG} is shown in Fig. ~\ref{fig:perturb 3}). The DIR estimates from ideal ($Y$) and unreliable measurements ($U$) are shown in Fig. ~\ref{fig:diCol}). We observe non-zero DIR estimates and add edges to $G_Z$ respectively. In particular, note the substantial rise in $I(u_1\to u_3\parallel u_2)$ and in $I(u_2\to u_3\parallel u_1)$. This indicates the presence of spurious links $1\to 3$ and $2\to 3$ in the inferred perturbed graph. 

\subsubsection{Multiple Perturbation}
\begin{figure}[t]
\centering
\subcaptionbox{\label{fig:H} 
    True generative graph.}
    {\begin{subfigure}{0.49\columnwidth}
    \centering
\begin{tikzpicture}[scale=0.65]
 \tikzstyle{vertex}=[circle,fill=none,minimum size=10pt,inner sep=0pt,thick,draw]
 \tikzstyle{pvertex}=[star,star points=14,fill=white,minimum size=10pt,inner sep=0pt,thick,draw]
  \node[vertex] (n1) {$1$};

  \node[vertex] (n2) at ($(n1)+ (4em,0em)$) {$4$};%

        \node[vertex] (n3) at ($(n1) - (0em,3.5em)$) {$2$};%

    \node[vertex] at ($(n2)-(0,3.5em)$) (n4) {$5$};%
  \node[vertex] at ($(n3)-(0,3.5em)$) (n5){$3$};
  \node[vertex] at ($(n4)-(0,3.5em)$) (n6){$6$};
  \draw[->,thick] (n1)--(n3);
  \draw[->,thick] (n2)--(n4);
  \draw[->,thick] (n4)--(n6);
  \draw[->,thick] (n3)--(n4);
  \draw[->,thick] (n3)--(n5);
  
\end{tikzpicture}
\end{subfigure}}
\subcaptionbox{
   \label{fig:pgH} 
    Network inferred from corrupt data-streams at nodes 2 and 5}
    {\begin{subfigure}{0.49\columnwidth}
    \centering
  \begin{tikzpicture}[scale=0.75]
 \tikzstyle{vertex}=[circle,fill=none,minimum size=10pt,inner sep=0pt,thick,draw]
 \tikzstyle{pvertexB}=[star,star points=10,minimum size=10pt,inner sep=0pt,thick,draw,fill=white]
  \tikzstyle{pvertexG}=[star,star points=10,minimum size=10pt,inner sep=0pt,thick,draw,fill=white]
  \node[vertex] (n1) {$1$};

  \node[vertex] (n2) at ($(n1)+ (4em,0em)$) {$4$};%

        \node[pvertexB] (n3) at ($(n1) - (0em,3.5em)$) {$2$};%

    \node[pvertexG] at ($(n2)-(0,3.5em)$) (n4) {$5$};%
  \node[vertex] at ($(n3)-(0,3.5em)$) (n5){$3$};
  \node[vertex] at ($(n4)-(0,3.5em)$) (n6){$6$};
  \draw[->,thick] (n1)--(n3);
  \draw[->,thick] (n2)--(n4);
  \draw[->,thick] (n4)--(n6);
  \draw[->,thick] (n3)--(n4);
  \draw[->,thick] (n3)--(n5);
  
\draw[-{stealth[scale=3.0]},red,thick,dashed] (n1) -- (n4);
\draw[-{stealth[scale=3.0]},red,thick,dashed] (n1) -- (n6);
\draw[-{stealth[scale=3.0]},red,thick,dashed] (n2)--(n3) ; 
\draw[-{stealth[scale=3.0]},red,thick,dashed] (n2) -- (n5)  ;
\draw[-{stealth[scale=3.0]},red,thick,dashed] (n4) -- (n5)  ;
\draw[-{stealth[scale=3.0]},red,thick,dashed] (n5) -- (n4)  ;
\draw[-{stealth[scale=3.0]},red,thick,dashed] (n5) -- (n6)  ;
\draw[-{stealth[scale=3.0]},red,thick,dashed] (n6) -- (n5)  ;
\draw[-{stealth[scale=3.0]},red,thick,dashed] (n3) -- (n6)  ;
\draw[-{stealth[scale=3.0]},red,thick,dashed] (n6) -- (n3)  ;
\draw[-{stealth[scale=3.0]},red,thick,dashed] (n5) to[out=40,in=320] (n3);
\draw[-{stealth[scale=3.0]},red,thick,dashed] (n6) to[out=40,in=320] (n4);
\draw[-{stealth[scale=3.0]},red,thick,dashed] (n4) to[out=120,in=30] (n3);
\draw[-{stealth[scale=3.0]},red,thick,dashed] (n1) to[out=220,in=130,bend right=75] (n5);
\draw[-{stealth[scale=3.0]},red,thick,dashed] (n2) to[out=220,in=130,bend left=75] (n6);
\end{tikzpicture}
 \end{subfigure}}
\subcaptionbox{
    \label{fig:di12} Comparison of directed information rate (DIR) estimates for links from nodes $1$ and $2$, between ideal data-streams $Y$ and uncertain measurements $U$. DIR $I$ is shown along X-axis and the sample length $n$ is along Y-axis.
  }{\begin{subfigure}{0.9\columnwidth}
  \centering
   \includegraphics[height=5cm,width=8cm]{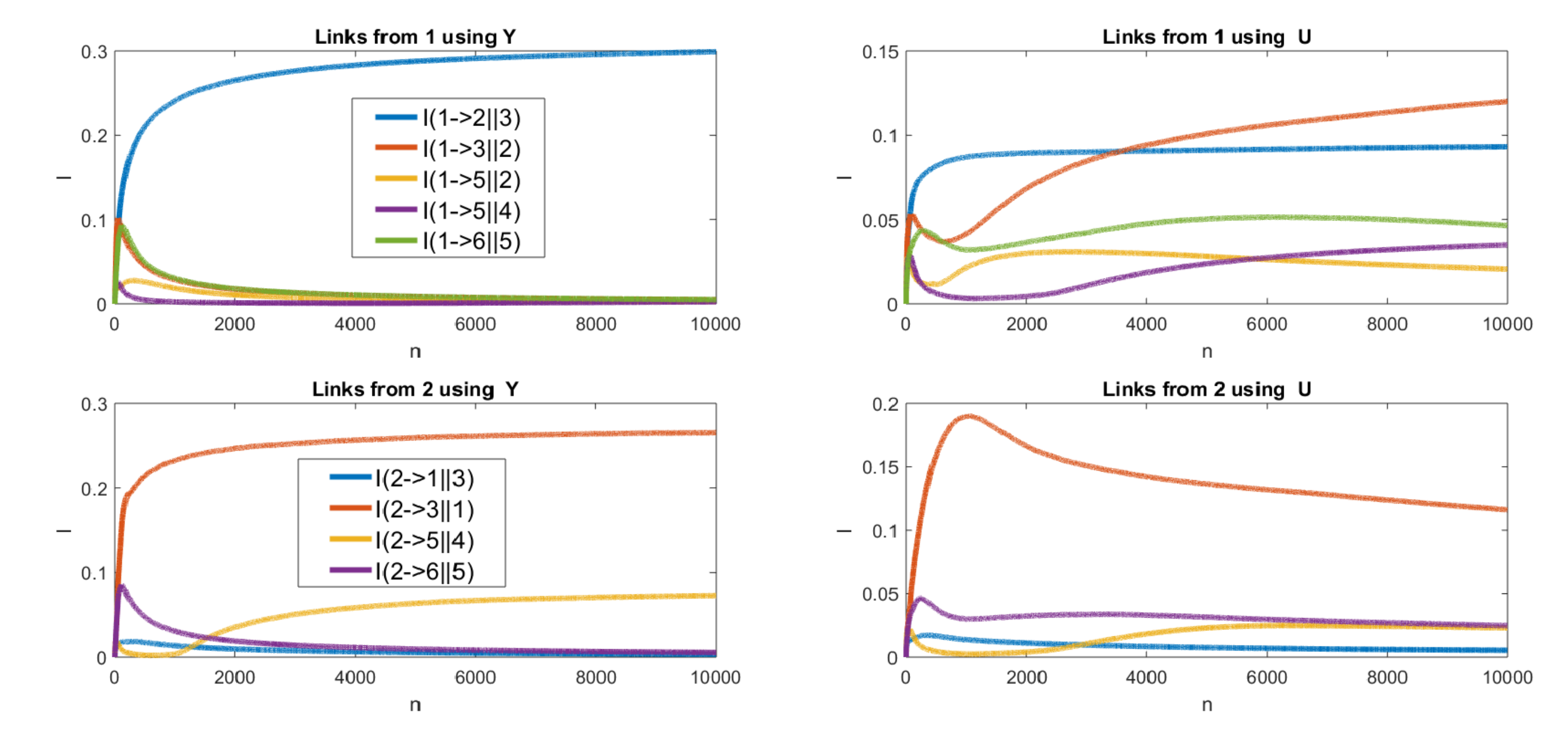}
    \end{subfigure}
  }
    \caption{\label{fig:trueHDI12}~\ref{fig:H}) shows true generative
      graph. ~\ref{fig:di12}) depicts DIR  estimates to detect links
      from nodes $1$ and $2$ using ideal measurements $Y$ and when
      there is corruption at nodes $2$ and $5$. It can be observed
      that many spurious links are detected. 
  }
\end{figure}
Consider a network of $6$ nodes as shown in Fig. ~\ref{fig:H}).  The dynamic interactions in the true generative model are as follows:
\begin{eqnarray*}
y_1[t]&=&e_1[t],\\
y_2[t]&=&y_1[t-1]+e_2[t],\\
y_3[t]&=&y_2[t-1]+e_3[t],\\
y_4[t]&=&e_4[t],\\
y_5[t]&=&(y_2[t-1]+y_4[t-1])\cdot e_5[t],\\
y_6[t]&=&y_5[t-1]+e_6[t]
\end{eqnarray*} where $e_1[t] \sim$ Bernouilli(0.55), $e_2[t] \sim$ Bernouilli(0.5), $e_3[t] \sim$ Bernouilli(0.2), $e_4[t] \sim$ Bernouilli(0.4), $e_5[t] \sim $ and $e_6[t] \sim$ Bernouilli(0.3)and `$+$' is logical `OR' operation while `$\cdot$' is logical `AND' operation. Each of $y_1[t],y_2[t],\dots ,y_6[t]$ has a finite alphabet $\{0,1\}.$
\begin{figure}[t]
\centering
\subcaptionbox{\label{fig:di34} A comparison of DIR estimates to detect links from nodes $3$ and $4$ using ideal measurements and when there is corruption at nodes $2$ and $5$ is shown. DIR $I$ is shown along X-axis and the sample length $n$ is along Y-axis. 
  }
{
\begin{minipage}{0.9\columnwidth}
\centering
   \includegraphics[height=4.9cm,width=8cm]{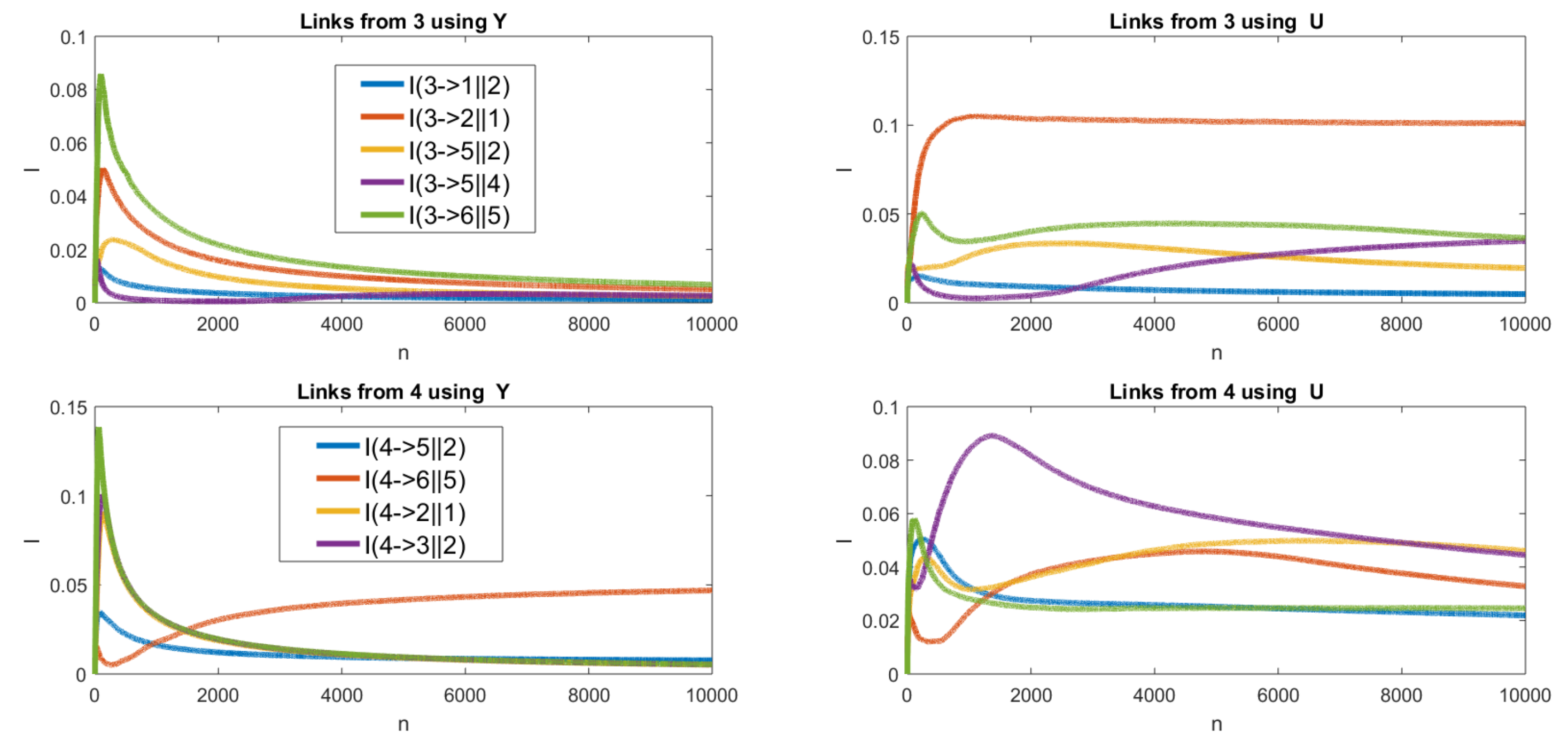}
\end{minipage}}
\subcaptionbox{\label{fig:di34} A comparison of DIR  estimates to detect links from nodes $5$ and $6$ using ideal measurements and when there is corruption at nodes $2$ and $5$ is shown. DIR $I$ is shown along X-axis and the sample length $n$ is along Y-axis.
  }{
\begin{minipage}{0.9\columnwidth}
\centering
   \includegraphics[height=5cm,width=8cm]{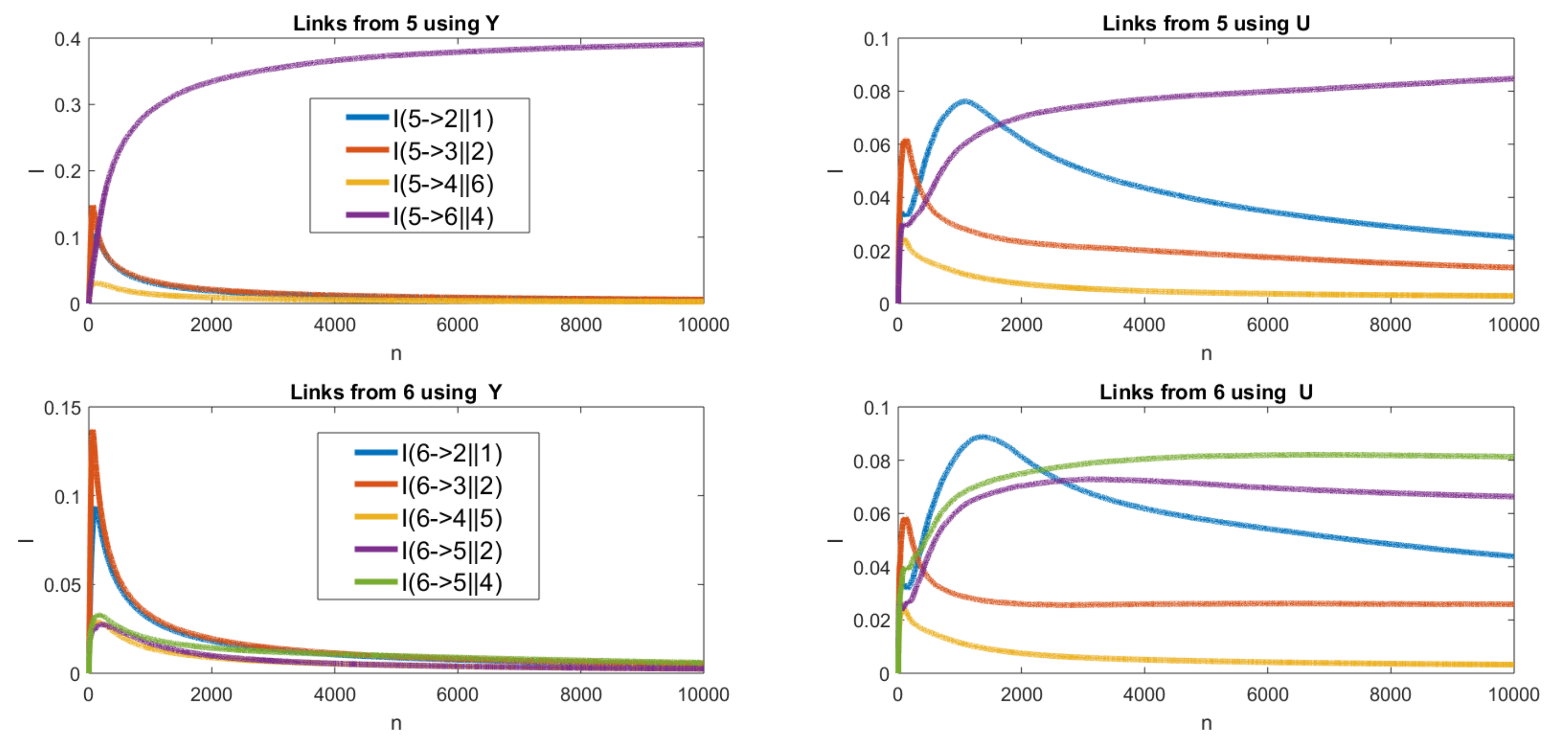}
\end{minipage}
}\caption{\label{fig:di3456}DI estimates to detect links from nodes 3,4,5 and 6. }
\end{figure}
The perturbations considered here are time-origin uncertainties at nodes 2 and 5. The corruption models takes the form:
\begin{equation*}
u_2[t] = \begin{cases}
  y_2[t-2], & \textrm{ with probability } 0.5 \\
  y_2[t], & \textrm{ with probability } 0.5.
\end{cases}
\end{equation*} and 
\begin{equation*}
u_5[t] = \begin{cases}
  y_5[t-2], & \textrm{ with probability } 0.5 \\
  y_5[t], & \textrm{ with probability } 0.5.
\end{cases}
\end{equation*} 
The perturbed graph predicted by Theorem ~\ref{thm:PG} is shown in figure  ~\ref{fig:pgH}). The DIR estimates from ideal ($Y$) and unreliable measurements ($U$) are shown in figures ~\ref{fig:di12}) and ~\ref{fig:di3456}). We observe non-zero DIR estimates and add edges to $G_Z$ respectively. For clarity of visualization, only non-zero DIR estimates that would be predicted by Theorem ~\ref{thm:PG} are shown.
\section{Conclusion}\label{sec:conclude}
We studied the problem of inferring directed graphs for
a large class of networks that admit non-linear and strictly
causal interactions between several agents. We provided necessary and sufficient conditions to determine the directed structure from corrupt data-streams. 
Doing so, we particularly
established that inferring causal structure from corrupt data-streams results
in spurious edges and we precisely characterized the
directionality of such spurious edges. Finally, we provided convergence results for the estimation of conditional directed information that was used to determine the directed structure. Simulation results were
provided to verify the theoretical predictions.  
\subsection*{Future Work}
Currently, the emphasis was on characterizing the effects of data corruption on network inference and determining how spurious probabilistic relations are introduced. Future work will
focus on quantifying the amount of data that is needed to
detect network inter-relationships using directed information. Moreover, the problem of removing spurious edges in the network
reconstructed from corrupt data streams will be addressed in
future. This will play an integral role
preceding system identification on networked systems. 
\appendices
\section{Proof for Theorem \ref{thm:PGrelaxed}}
\label{appx:relax}
Suppose $i\rightarrow j$ is in  $A_Z$. Then there is a trail, $trl_G$, described by 
	$i=v_1-v_2-\cdots-v_k=j$ in $G$ satisfying conditions in Definition \ref{def:pg}. We will first construct a trail in the perturbed DBN, $G'_Z,$ from a node in $w_i^{(t-1)}$ to $w_j[t]$ for some $t>0$. We can construct a trail in $G'_Z$ as follows: for all $l\in \{1,2,\ldots , k-1\}$, set $t_{l}=t_{l+1}-k_{v_{l+1}v_{l}}$ if $v_l\to v_{l+1} $ holds in $trlG$. Otherwise, set $t_{l}=t_{l+1}+k_{v_{l}v_{l+1}}$ if $v_l\gets v_{l+1} $ holds in $trlG$. Such a construction is feasible because by condition C~\ref{c1}), numbers $k_{v_{l+1}v_{l}}$ and $k_{v_{l}v_{l+1}}$ exists for all $l\in \{1,2,\ldots , k-1\}$ and at all times. Thus, we have a trail $y_i[t_1]-y_{v_{2}}[t_2]-y_{v_3}[t_3]-\ldots-y_{v_{k-1}}[t_{k-1}]-y_j[t_k]$. For all $m\in \{1,2,\ldots ,k\}$ if $v_m\in Z$, there exists a number $k_m>0$ following conditions C~\ref{c2}). If B~\ref{c4}) also holds, then $k_m\ge 0$. Let $t>\max\{t_1,\ldots, t_{k-1}\}$, and for all $m\in \{1,2,\ldots ,k\}$ if $v_m\in Z$, let $t>t_m+k_m$ also hold. Depending on whether $i$ or $j$ is a perturbed node, we have four cases on either end of the above trail. 
\begin{enumerate}[A)]
\item Consider the case $i,j\in Z$. As $i\in Z$, using condition C~\ref{c2}) $u_i[t_1+k_i]\leftarrow y_i[t_1]$ holds true. Choose $t$ sufficiently large so that $t>t_1+k_i$ also holds. As $j\in Z$, using C~\ref{c2}), $t$ can be sufficiently large so that we have $y_j[t_k]\rightarrow u_j[t]$ where $t=t_k+k_j$ and $k_j\ge 1$. If B~\ref{c3}) holds, then we can choose $t$ sufficiently large such that at the end of the trail we take $s$ steps from $y_j[t_k]$ to $u_j[t]$ such that the tail is of the form $y_j[t_k]\to y_j[t_k+k'_j]\to \cdots \to y_j[t_k+sk'_j]\to u_j[t]$ with $t=t_k+sk'_j+k_j$. Thus, the constructed trail in $G'_Z$ is either 
	$w_i[t_1+k_i]=u_i[t_1+k_i]\leftarrow y_i[t_1]-y_{v_{2}}[t_2]-y_{v_3}[t_3]-\cdots-y_{v_{k-1}}[t_{k-1}]-y_j[t_k]\rightarrow u_j[t]=w_j[t]$, or $w_i[t_1+k_i]=u_i[t_1+k_i]\leftarrow y_i[t_1]-y_{v_{2}}[t_2]-y_{v_3}[t_3]-\cdots-y_{v_{k-1}}[t_{k-1}]-y_j[t_k]\to y_j[t_k+k'_j]\to \cdots \to y_j[t_k+sk'_j]\to u_j[t]=w_j[t]$ with $t>\max\{t_1+k_i,t_1,\ldots, t_k,\ldots ,t_k+sk'_j\}$, and for all $m\in \{1,2,\ldots ,k\}$ if $v_m\in Z$, $t>t_m+k_m$. 
	\item  Consider the case $i\in Z$ but $j\not \in Z$. Choose $t$ as $t_k$. As $i\in Z$, using condition C~\ref{c2}) $u_i[t_1+k_i]\leftarrow y_i[t_1]$ holds true. Choose $t$ sufficiently large so that $t>t_1+k_i$ also holds. Thus, we have constructed a trail in $G'_Z$ which is of the form: 
	$w_i[t_1+k_i]=u_i[t_1+k_i]\leftarrow y_i[t_1]-y_{v_{2}}[t_2]-y_{v_3}[t_3]-\cdots-y_{v_{k-1}}[t_{k-1}]- y_j[t]=w_j[t]$ with $t>\max\{t_1+k_i,t_1,\ldots, t_{k-1}\}$, and for all $m\in \{1,2,\ldots, k\}$ if $v_m\in Z$, $t>t_m+k_m$.
	\item Consider the case $i\not \in Z$ but $j \in Z$. Following arguments presented in case (A)  we conclude that the constructed trail of form $ w_i[t_1]=y_i[t_1]-y_{v_{2}}[t_2]-y_{v_3}[t_3]-\cdots-y_{v_{k-1}}[t_{k-1}]-y_j[t_k]\rightarrow u_j[t]=w_j[t]$, or of form $ w_i[t_1]=y_i[t_1]-y_{v_{2}}[t_2]-y_{v_3}[t_3]-\cdots-y_{v_{k-1}}[t_{k-1}]-y_j[t_k]\rightarrow y_j[t_k+k'_j]\to \cdots \to y_j[t_k+sk'_j]\to u_j[t]=w_j[t]$ exists in the perturbed DBN $G'_Z$ with $t>\max\{t_1,\ldots, t_k,\ldots ,t_k+sk'_j\}$, and for all $m\in \{1,2,\ldots ,k\}$ if $v_m\in Z$, $t>t_m+k_m$.
	\item  Consider the case $i\not \in Z$ and  $j \not  \in Z$. Following arguments presented in Case (B)  we conclude that the trail  $ w_i[t_1]=y_i[t_1]-y_{v_{2}}[t_2]-y_{v_3}[t_3]-\cdots-y_{v_{k-1}}[t_{k-1}]-y_j[t]=w_j[t]$ exists in the perturbed DBN $G'_Z$ with $t>\max\{t_1,\ldots, t_{k-1}\}$, and for all $m\in \{1,2,\ldots ,k\}$ if $v_m\in Z$, $t>t_m+k_m$.
	\end{enumerate}	
	\noindent We will now argue that in each of the cases above, the constructed trail is active given $\theta :=\{w_j^{(t-1)}, \mathcal{W}_{\bar{j}\bar{i}}^{(t-1)}\}.$
	
	{\it Sub-trails with colliders:} For all the trails in $G'_Z$ constructed under various cases above consider a sub-trail of the form $y_{v_{m-1}}[t_{m-1}]\rightarrow y_{v_{m}}[t_{m}] \leftarrow y_{v_{m+1}}[t_{m+1}] $. Clearly, $v_m$  cannot be either $i$ or $j.$ If $v_m\not\in Z$ then  as  $t_m<t$, we have   $y_{v_{m}}[t_{m}]\in 
        w_{\bar{j}\bar{i}}^{(t-1)}$ and thus the sub-trail is active. If  $v_m \in Z$ then the  corrupted version of $y_{v_{m}}[t_{m}]$ is $ u_{v_{m}}[t_{m}+k_{v_m}]=w_{v_m}[t_m+k_{v_m}]$ and as $t_m+k_{v_m}<t$, we have $w_{v_{m}}[t_{m}+k_m] \in  w_{\bar{j}\bar{i}}^{(t-1)}$. Thus the collider  $y_{v_{m}}[t_{m}]$ has a descendant $w_{v_m}[t_m+k_{v_m}] \in \theta $. Thus the sub-trail remains active. Thus no collider can deactivate the trails in $G'_Z.$ 
	
	{\it Sub-trails with with no colliders:} Now consider any node
        $y_{v_m}[t_m]$ which is not a collider. Note that in the
        trails for the cases (A), (B), (C), and (D), $y_j$ and $y_i$
        can only appear  as an intermediate node only if  they are
        corrupted. In such cases, neither $y_i[t_1]$ nor $y_j[t_k]$
        belong to $\theta .$ Thus, if $y_j$ or $y_i$ are intermediate nodes, they cannot deactivate the trails given $\theta .$ 
	Consider an intermediate node $v_m\not \in \{i,j\}.$ From
        Definition ~\ref{def:pg}P~\ref{case:notCollider}), $v_m$ is corrupted. Thus
        $y_{v_m}[t_m]\not =w_{v_m}[t_m] $ and  $y_{v_m}[t_m]$ cannot
        deactivate the trail as $y_{v_m}[t_m]\not \in \theta .$ 
        \hfill\qed
\section{Proof for Theorem \ref{thm:convergeConditionalI}}
\label{appx:conditional}
To prove the theorem, we require two results from \cite{jiao2013universal}. The following lemma shows that with sufficiently large data, the conditional probability assignment by CTW converges to the true probability assignment for a Markov process.
\begin{lemma}\label{lem:convergeQ}
Let $Q$ be the probability assignment described by CTW. Let $X$ be a stationary and finite alphabet Markov process with finite Markov order which is bounded by the prescribed tree depth of CTW algorithm. Let $P$ be the true probability for $x$. Then,
\begin{equation}
\lim _{n\to \infty}Q(x[n]\mid x^{(n-1)})-P(x[n]\mid x^{(n-1)})=0 \qquad \mbox{P-a.s}.
\end{equation}
\end{lemma}
Next, we will later use the following proposition which is a rephrased result from \cite{jiao2013universal}.
\begin{proposition}\label{prop:Hconverge}
Let $Q$ be the probability assignment in the CTW algorithm. Suppose, $X,Y$ are jointly stationary irreducible aperiodic finite-alphabet Markov processes whose order is bounded by the prescribed tree depth of the CTW algorithm. Let $\hat H(Y^{(n)}\parallel X^{(n)})=-\frac{1}{n}\sum _{i=1}^n\sum _{y_i}Q(y[i]\mid X^{(i-1)},Y^{(i-1)})\cdot  \log \frac{1}{Q(y[i]\mid X^{(i-1)},Y^{(i-1)})}$. Then,
\begin{equation}\label{eq:Hconverge}
\lim _{n\to \infty} \hat H(Y^{(n)}\parallel X^{(n)})-H_r(Y\parallel X)=0 \qquad \mbox{P-a.s},
\end{equation}
\end{proposition}   
Recall the expression for the conditional DI estimator from \eqref{eq:conditionalDI}:
\begin{multline}\label{eq:conditionalDI_appx}
\hat I(X^{(n)}\to Y^{(n)}\parallel Z^{(n)}) =\\ \frac{1}{n}\sum _{i=1}^n\sum _{y[i]}Q(y[i]\mid X^{(i-1)},Y^{(i-1)},Z^{(i-1)})\cdot \\  \log \frac{1}{Q(y[i]\mid Y^{(i-1)},Z^{(i-1)})}
\\  -\frac{1}{n}\sum _{i=1}^n\sum _{y[i]}Q(y[i]\mid X^{(i-1)},Y^{(i-1)},Z^{(i-1)})\cdot \\ \log \frac{1}{Q(y[i]\mid X^{(i-1)},Y^{(i-1)},Z^{(i-1)})}
\end{multline}
We will show that the first term(call it T1) in equation \eqref{eq:conditionalDI_appx} converges to $H_r(Y\parallel Z)$ and the second term (call it T2) in \eqref{eq:conditionalDI_appx} converges to $H_r(Y\parallel X,Z)$.

Convergence of T2: Let $V=\{X,Z\}$. Thus, T2 can be written as $\hat H(Y^{(n)}\parallel V^{(n)})=-\frac{1}{n}\sum _{i=1}^n\sum _{y[i]}Q(y[i]\mid V^{(i-1)},Y^{(i-1)})\cdot  \log \frac{1}{Q(y[i]\mid V^{(i-1)},Y^{(i-1)})}$. Using, proposition \ref{prop:Hconverge}, we thus have that $\lim _{n\to \infty} \hat H(Y^{(n)}\parallel V^{(n)}) \to H_r(Y\parallel V)$ almost surely.

Convergence of T1: Subtract $H_r(Y\parallel Z)$ from $T1$ and express $T1-H_r(Y\parallel Z)=F_n +S_n$ where,
\begin{multline}\label{eq:F_n}
F_n=\frac{1}{n}\sum _{i=1}^n\sum _{y[i]}P(y[i]\mid X^{(i-1)},Y^{(i-1)},Z^{(i-1)})\cdot \\  \log {P(y[i]\mid Y^{(i-1)},Z^{(i-1)})}
\\  -\frac{1}{n}\sum _{i=1}^n\sum _{y[i]}Q(y[i]\mid X^{(i-1)},Y^{(i-1)},Z^{(i-1)})\cdot \\ \log {Q(y[i]\mid Y^{(i-1)},Z^{(i-1)})},
\end{multline} 
\begin{multline}\label{eq:S_n}
S_n=-\frac{1}{n}\sum _{i=1}^n\sum _{y[i]}P(y[i]\mid X^{(i-1)},Y^{(i-1)},Z^{(i-1)})\cdot \\  \log {P(y[i]\mid Y^{(i-1)},Z^{(i-1)})} -H_r(Y\parallel Z)
\end{multline} 
By ergodicity, $S_n$ converges to zero almost surely. We need to show that $F_n$ converges to zero almost surely. Rewrite $F_n=\frac{1}{n}\sum _{i=1}^n\beta _i$ where, 
\begin{multline}
\beta _i= \sum _{y[i]}P(y[i]\mid X^{(i-1)},Y^{(i-1)},Z^{(i-1)})\cdot \\ \log {P(y[i]\mid Y^{(i-1)},Z^{(i-1)})}
\\  -\sum _{y[i]}Q(y[i]\mid X^{(i-1)},Y^{(i-1)},Z^{(i-1)})\cdot\\ \log {Q(y[i]\mid Y^{(i-1)},Z^{(i-1)})}
\end{multline}
By Lemma \ref{lem:convergeQ}, the CTW probabilities $Q(y[i]\mid X^{(i-1)},Y^{(i-1)},Z^{(i-1)})$ converges to true probabilities $P(y[i]\mid X^{(i-1)},Y^{(i-1)},Z^{(i-1)})$ almost surely. Therefore,
\begin{equation}
\lim _{i\to \infty}\beta _i=0 \qquad \mbox{P-a.s.}
\end{equation}
Hence, by Cesaro mean \cite{cover2012elements} we have:
\begin{equation}
\lim _{n\to \infty}F_n =\lim _{n\to \infty}\frac{1}{n}\beta _i=0 \quad\mbox{P-a.s.} \hfill\qed
\end{equation}
\bibliographystyle{IEEEtran}
\bibliography{ref}   
\begin{IEEEbiography} [{\includegraphics[width=1in,height=1.25in,clip,keepaspectratio]{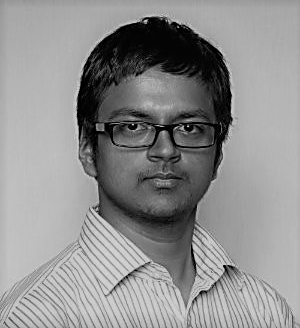}}]{Venkat Ram Subramanian}
received the B.Tech degree in electrical engineering from SRM University, Chennai, India, in 2014, and the M.S. degree in electrical engineering from the University of Minnesota, Minneapolis, in 2016. 
Currently, he is working towards a Ph.D. degree at the University of Minnesota. His Ph.D. research is on learning dynamic relations in networks from corrupt data-streams. In addition to system identification and stochastic systems, his research interests also include grid modernization and optimal energy management in Distributed Energy Resources (DER).  
\end{IEEEbiography}
\begin{IEEEbiography}[{\includegraphics[width=1in,height=1.25in,clip,keepaspectratio]{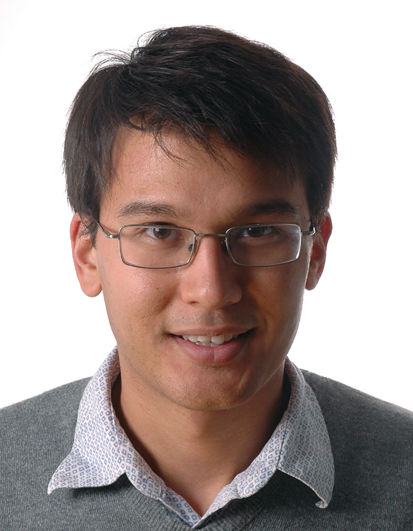}}] {Andrew Lamperski}
(S'05--M'11) received the B.S. degree in biomedical engineering and
mathematics in 2004 from the Johns Hopkins University, Baltimore, MD,
and the Ph.D. degree in control and dynamical systems in 2011 from the
California Institute of Technology, Pasadena. He held postdoctoral
positions in control and dynamical systems at the California Institute
of Technology from 2011--2012 and in mechanical engineering at The
Johns Hopkins University in 2012. From 2012--2014,
did
postdoctoral work in the Department of Engineering, University of
Cambridge, on a scholarship from the Whitaker International
Program. In 2014, he joined the Department of Electrical and Computer
Engineering, University of Minnesota as an Assistant Professor. His
research interests include optimal control, optimization, and identification, with applications to neuroscience and robotics.
\end{IEEEbiography}
\begin{IEEEbiography}[{\includegraphics[width=1in,height=1.25in,clip,keepaspectratio]{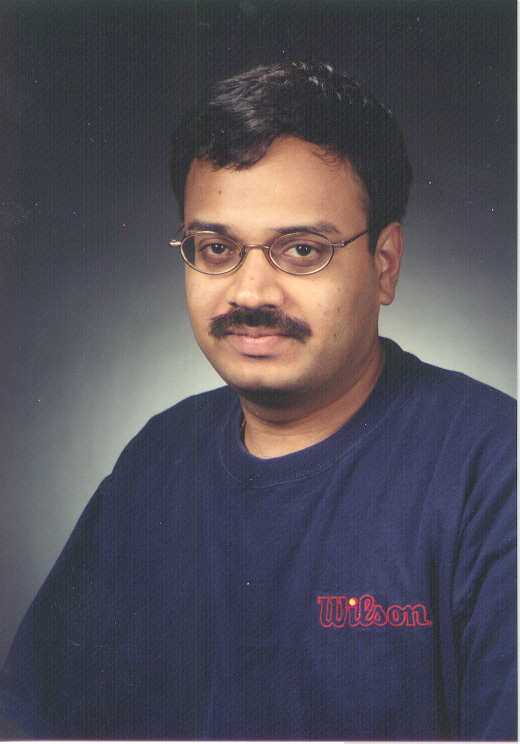}}] {Murti Salapaka} (SM'01--F'19)
  Murti Salapaka received the bachelor’s degree from the Indian 
Institute of Technology, Madras, India, in 1991, and the Master’s and 
Ph.D. degrees from the University of California, Santa Barbara, CA, USA, 
in 1993 and 1997, respectively, all in mechanical engineering. He was 
with Electrical Engineering department, Iowa State University, from 1997 
to 2007. He is currently the Vincentine Hermes-Luh Chair Professor with 
the Electrical and Computer Engineering Department, University of 
Minnesota, Minneapolis, MN, USA. Prof. Salapaka was the recipient of the 
NSF CAREER Award and the ISU—Young Engineering Faculty Research Award 
for the years 1998 and 2001, respectively. He is an IEEE Fellow.
\end{IEEEbiography}
\end{document}